\documentclass{article}
\usepackage{amsmath,amsthm,amsfonts,amssymb}
\usepackage[utf8]{inputenc}
\usepackage[english]{babel}
\usepackage{comment}
\usepackage{url}
\usepackage{graphicx, subfigure}
\usepackage[pdftex]{hyperref}
\usepackage[usenames,dvipsnames]{color}
\usepackage{longtable}
\usepackage{csquotes}
\usepackage{float}
\usepackage{tikz}
\usepackage{accents}
\usepackage{xcolor}
\usepackage{mathtools}
\usepackage{multicol}
\usepackage{caption, subcaption}
\usepackage{xcolor}
\usetikzlibrary{decorations.pathmorphing, shadows, decorations.markings}
\usepackage{graphicx}
\newtheorem{theorem}{Theorem}[section]

\newtheorem{proposition}[theorem]{Proposition}

\newtheorem{lemma}[theorem]{Lemma}
\newtheorem{definition}[theorem]{Definition}

\newtheorem{remark}[theorem]{Remark}
\usepackage{xcolor}

\renewcommand{\Im}{\mathrm{Im}}
\renewcommand{\Re}{\mathrm{Re}}
\newcommand\at[2]{ #1  \Bigg |_{#2}}

\newcommand{\innerproduct}[2]{\langle #1, #2 \rangle}

\title{Reconstructing resonant phase oscillator interactions from noisy time series}
\author{Bengi D\"onmez\thanks{
 Department of Mathematics, Vrije Universiteit Amsterdam, Amsterdam, the Netherlands, {\tt b.donmez@vu.nl} 
 }, Bob Rink\thanks{
 Department of Mathematics, Vrije Universiteit Amsterdam, Amsterdam, the Netherlands, {\tt b.w.rink@vu.nl}
}}

\date{\today}

\begin{document}

\maketitle
\abstract{We present a novel method for reconstructing networks of coupled phase oscillators from noisy time series. Noise and uncertainty can make it hard or impossible to distinguish different oscillator networks based on observed dynamical behavior. Thus, our method does not aim to determine exact phase equations for the oscillators, but instead recovers their first and second order resonant normal form. This normal form contains crucial information on the underlying network and yields accurate approximations of the dynamics. We provide rigorous estimates on the accuracy of the reconstructed normal form, and we illustrate the method with numerical examples.}
\section{Introduction}\label{sec:intro}
Networks of coupled oscillators arise abundantly in physics, biology, and engineering, and they are often modeled by coupled differential equations. The goal of network reconstruction methods is to infer the coupling topology of such networks from observations of solutions to these differential equations. Many of these methods work by first extracting phase signals from oscillatory time series generated by the network, and subsequently fitting a phase vector field to these phase signals. 
Unfortunately, the resulting reconstructed phase equations typically do not allow us to infer the full coupling topology of the original network. One reason is that distinct oscillator networks may generate very similar dynamics, which may become indistinguishable when the observations of the systems are limited, noisy, or ill-conditioned. However, even in the absence of noise and uncertainty in the observations, the relation between the original oscillator network and the equations of motion for its phases is ambiguous, due to the intrinsic non-uniqueness of phase variables and phase extraction methods. As a result of this non-uniqueness, the same oscillator network may faithfully be described by multiple distinct phase equations, while the same phase equations may govern a variety of oscillator networks.  This ambiguity forms a fundamental obstruction to the reliable reconstruction of an oscillator network from  observations of its dynamics.

In this paper, we address this problem by presenting a new reconstruction method 
for networks of weakly coupled phase oscillators. Rather than attempting to reconstruct the exact equations of motion of such networks from observations, our method reconstructs their so-called resonant normal form. Networks with the same normal form generate nearly identical dynamics over long timescales, and are thus indistinguishable in the presence of noise or uncertainty. On the contrary, networks with different normal forms typically display quantitatively different dynamical behavior.  
 This makes it more feasible to reconstruct the normal form than the true equations of motion, and it allows us to quantify the error in the reconstructed normal form.
The normal form also provides reliable predictions of the long-term dynamics of the true network, and  it encodes structural properties of the true network, such as symmetries.
\subsubsection*{Problem setting}
Concretely, we consider interacting phase oscillator systems of the form 
\begin{align}\label{equationsofmotionintro}
\dot \phi_j = \omega_j + \varepsilon F_j^{(1)}(\phi)  +   \varepsilon^{2} F_j^{(2)}(\phi) + \varepsilon^3 R_j^{(3)}(\phi,\varepsilon) \  \mbox{for} \  1 \leq j \leq n\, ,
\end{align} 
where $\phi_j\in \mathbb{T}:=\mathbb{R}/2\pi\mathbb{Z}$ is a phase variable, $\varepsilon\geq 0$ is a small coupling parameter, and the $F_j^{(l)}$ smooth interaction functions  with  Fourier expansions 
\begin{align}\label{eq:Fexpansion}
\ F_j^{(l)}(\phi) = \sum_{k\in \mathbb{Z}^n} A_{j, k}^{(l)} e^{i\langle k, \phi \rangle }\ \mbox{with}\ \overline{A_{j,k}^{(l)}} = A_{j,-k}^{(l)}. 
\end{align}
 Our goal is to reconstruct (some of) the Fourier coefficients $A_{j,k}^{(l)}$ from finite time observations $\Phi^{(m)}(t)$ (with $1\leq m \leq M$ and $0\leq t\leq T$) of true solutions $\phi^{(m)}(t)$ to \eqref{equationsofmotionintro}. We assume that the uncertainty in these observations is small and bounded, that is,
$$| \phi^{(m)}_j(t) - \Phi^{(m)}_j(t)| \leq L_j\varepsilon \ \mbox{for all} \ 1\leq j\leq n, 1\leq m\leq M \ \mbox{and} \ 0 \leq t \leq T\, , $$
and for some uniform constants $L_1, \ldots, L_n>0$.   
We will argue in Section \ref{sec:framework} that this is the natural assumption to make when  \eqref{equationsofmotionintro} describes the dynamics on an invariant torus (or ``phase reduction'') of a higher-dimensional coupled oscillator system, and when the phase signals $\Phi^{(m)}(t)$ are extracted from observations of these high-dimensional oscillators. 

Network reconstruction methods often work by fitting a nonlinear phase vector field of the form \eqref{equationsofmotionintro} to the observations $\Phi^{(m)}(t)$, for instance, in the form of a finite Fourier series 
\begin{equation}\label{fittedvf}
\dot \phi_j \approx \omega_j + \varepsilon \sum_{k\in \mathcal{K}}\hat A_{j,k}e^{i\langle k,\phi\rangle}\, ,
\end{equation} 
with $\mathcal{K}\subset \mathbb{Z}^n$ a (large) finite ``library'' of pre-selected Fourier labels.  The simplest option is perhaps to choose the $\hat A_{j,k}$ to be minimizers of the sum of squares 
\begin{equation}\label{penalization}
B\mapsto \sum_{j=1}^M \left| \dot \Phi^{(m)}(0) - \omega_j - \varepsilon \sum_{k\in \mathcal{K}} B_{j,k} e^{i\langle k,\Phi^{(m)}(0)\rangle} \right|^2 \, .
\end{equation}
An alternative to fitting a phase vector field to  observed phase velocities as described above, is to fit a time-$T$ map to the finite-time phase drifts $\Phi^{(m)}(T)-\Phi^{(m)}(0)-\omega_jT$. This idea appears  for instance in \cite{matsuki_network_2024, rosenblum_identification_2002}, and it avoids the use of instantaneous phase velocities -- which can be hard to determine accurately in practice. 
We present a related approach in this paper, which works by fitting ``resonant'' time-$T$ maps. Specifically, we present  a first order and a second order resonant reconstruction method. 
\subsubsection*{First order resonant reconstruction}
Our first order method relies on Lemma \ref{lem:approxlemma} below, which states that, under  mild conditions on \eqref{equationsofmotionintro} and \eqref{eq:Fexpansion}, 
any exact solution to \eqref{equationsofmotionintro} satisfies 
\begin{align}\label{firstorderapproxresonant} 
 \phi_j(T) \! - \! \phi_j(0) \!-\! \omega_j T = \varepsilon T \!\! \sum_{\langle \omega,k\rangle = 0} \!\!  A_{j,k}^{(1)} e^{i\langle k, \phi(0)\rangle}  + \mathcal{O}(\varepsilon)\, \ \mbox{when}\  T \sim \frac{1}{\sqrt{\varepsilon}}.
 \end{align}
Thus, apart from a small correction of the order $\varepsilon$,  over timescales of the order $T\sim  \frac{1}{\sqrt{\varepsilon}}$, every solution to \eqref{equationsofmotionintro} is governed by a slow linear drift of the order $\varepsilon T \sim \sqrt{\varepsilon}$ away from the unperturbed motion  $\phi_j(T)=\phi_j(0)+\omega_jT$. It is important to remark that this linear phase drift is determined by so-called {\it resonant} terms only, that is, the terms $A_{j,k}^{(1)} e^{i\langle k, \phi\rangle}$ in \eqref{eq:Fexpansion} with $\langle \omega, k\rangle =0$. 
Nonresonant terms (those for which $\langle \omega, k\rangle \neq 0$) do not contribute to the linear phase drift.  Their contribution to the dynamics is instead contained in the correction term of the order $\varepsilon$. 

Inspired by this observation, our first order resonant reconstruction method first  selects a finite library 
$$\mathcal{K} \subset \{ k\in \mathbb{Z}^n \, |\, \langle \omega, k\rangle = 0 \}$$ 
of resonant Fourier labels, fixes a time $T$ of the order $\frac{1}{\sqrt{\varepsilon}}$, and then chooses the coefficients $\hat A_{j,k}$  in \eqref{fittedvf} to be minimizers of the sum of squares
$$B\mapsto \sum_{m=1}^M \left|\Phi_j^{(m)}(T) - \Phi_j^{(m)}(0) - \omega_j T  - \varepsilon T\sum_{k\in \mathcal{K}} B_{j,k} e^{i\langle k,\Phi^{(m)}(0) \rangle }\right|^2 \, .
$$ 
This procedure yields estimates for the first order {\it resonant} Fourier coefficients in \eqref{equationsofmotionintro}, and it does not attempt to estimate the nonresonant coefficients. The reason is that the contribution of the nonresonant terms to the dynamics is of the order $\varepsilon$. It is therefore of the same order of magnitude as the uncertainty in the observed phase drifts  
$$| ( \Phi_j^{(m)}(T)-\Phi_j^{(m)}(0)-\omega_jT) - ( \phi_j^{(m)}(T)-\phi_j^{(m)}(0)-\omega_jT) | \sim  \varepsilon\, .$$ 
 The nonresonant Fourier coefficients in \eqref{eq:Fexpansion} therefore cannot realistically be inferred from the observations $\Phi^{(m)}(t)$,  and any attempt to estimate them would  lead to an  underdetermined or ill-conditioned regression problem. 
On the contrary, the estimators $\hat A_{j,k}$ provided by our resonant reconstruction method can often be guaranteed to lie close to the true first order resonant Fourier coefficients $A_{j,k}^{(1)}$.  In particular, Theorem \ref{general_thm} roughly speaking states that  
$$A^{(1)}_{j,k} - \hat A_{j,k} \sim  \sqrt{\varepsilon}\ \mbox{for all} \ 1\leq j\leq n \ \mbox{and} \ k\in\mathcal{K}\, ,$$
under certain verifiable conditions on the observations $\Phi^{(m)}(t)$.
The presence of bounded noise makes it impossible to derive similar rigorous bounds for estimators of nonresonant Fourier coefficients. 

To illustrate the strength of our first order resonant reconstruction method, we refer to Figures \ref{fig:no_noise50}, \ref{fig:weak_noise50} and \ref{fig:strong_noise50}, in which we compare the performance of our resonant method with a reconstruction method that fits a time $T$-map with both resonant and nonresonant terms to observations of finite time phase drifts. The figures show that resonant reconstruction is consistently reliable in the presence of observational noise, while the more general reconstruction method starts to fail as the noise level is increased.

\begin{figure}[t]
 \centering
   \includegraphics[width=.99\linewidth]{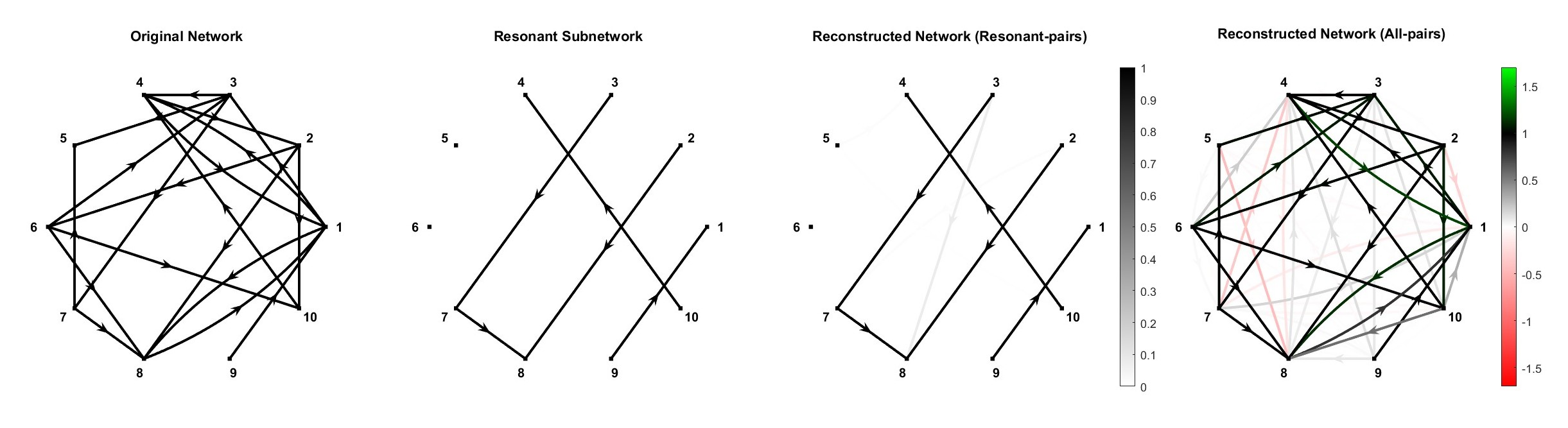}
   \vspace{-.5cm}
\caption{\footnotesize{Comparison of resonant and nonresonant network reconstruction methods in the absence of observational noise. We consider a directed Erdős-Rényi graph with $n=10$ nodes, where a directed edge probability $p=\frac{1}{5}$ was used to generate the graph. Phase dynamics is defined for each node of this graph by randomly choosing an intrinsic frequency $\omega_i\in \{1,2,3\}$ for each node, fixing the coupling strength $\varepsilon = 0.01$, and coupling the nodes by sinusoidal interaction functions of the form $\varepsilon a_{ij} \sin(\phi_i-\phi_j)$, where $a_{ij}\in \{0,1\}$ are the coefficients of the adjacency matrix of the graph. The resulting phase dynamics was integrated using \texttt{ode89} in \textsc{Matlab} over the time interval $[0,T] = [0, \varepsilon^{-\frac{1}{2}}] = [0,10]$ for a total of $M=50$ initial conditions, chosen independently from the uniform distribution on $[0,2\pi]^{10}$.  
Time-$T$ maps were then fitted to the phase drifts observed in the resulting time series, using two different  libraries: a resonant library containing the functions  $\sin(\phi_i-\phi_j)$ only when $\omega_i=\omega_j$ and $i\neq j$, and an all-pairs library containing $\sin(\phi_i-\phi_j)$ for all $i\neq j$. From left to right: original network with $a_{ij} \in \{0,1\}$; resonant subnetwork consisting of the connections in the original network with $\omega_i=\omega_j$; reconstruction using the resonant library; reconstruction using the all-pairs library. In this noise-free setting, the all-pairs reconstruction provides an excellent approximation of the original network: the two networks are visually almost indistinguishable. }}
   \label{fig:no_noise50}
 \end{figure}
 \begin{figure}[h]
 \centering
   \includegraphics[width=.99\linewidth]{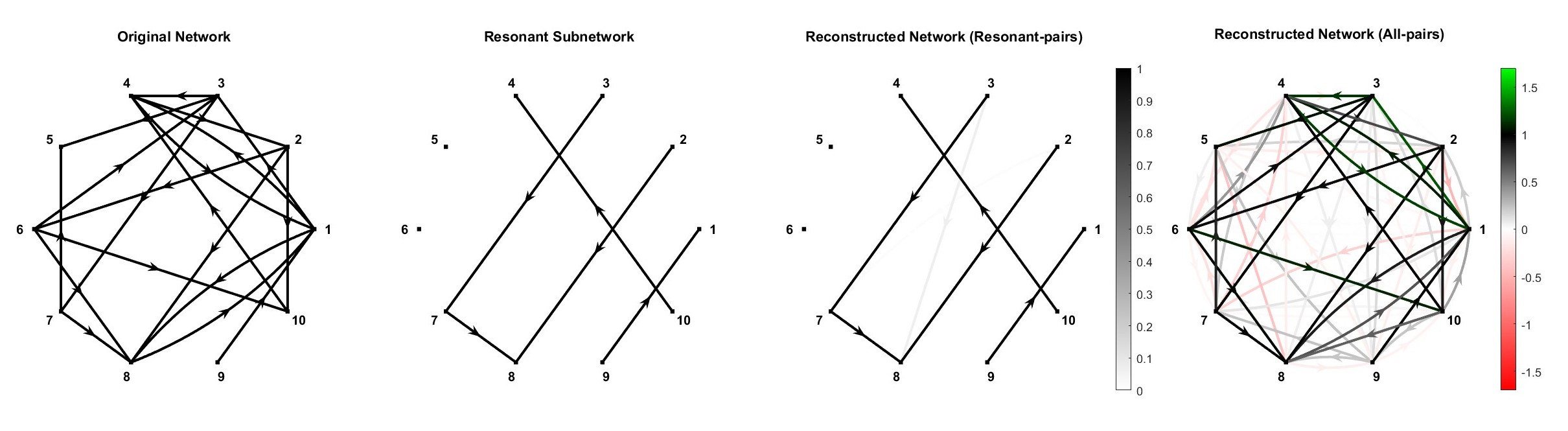}
    \vspace{-.5cm}
\caption{\footnotesize{Comparison of resonant and nonresonant network reconstruction methods with small observational noise. The original network is the same as in Figure \ref{fig:no_noise50} but the integrated time series were perturbed by i.i.d.\ random variables uniformly distributed on $[-\varepsilon,\varepsilon]$ (mean $0$, variance $\varepsilon^{2}/3$) before the reconstruction step. The resonant reconstruction remains a good approximation of the resonant subnetwork, whereas the all-pairs reconstruction already shows some discrepancies from the original network.}}
   \label{fig:weak_noise50}
 \end{figure}
 \begin{figure}[h]
 \centering
   \includegraphics[width=.99\linewidth]{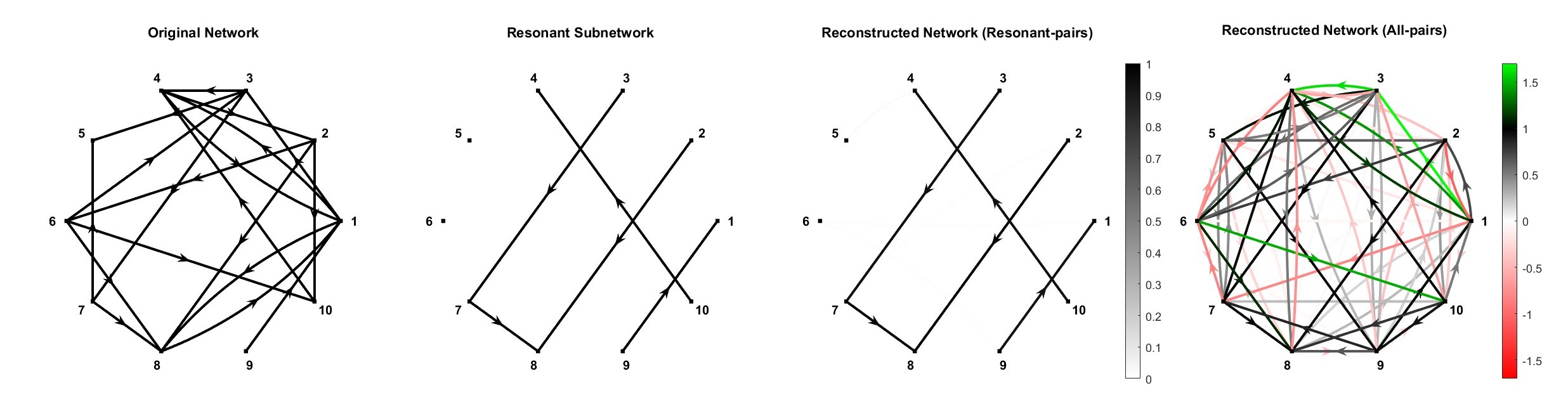}
    \vspace{-.5cm}
\caption{\footnotesize{Comparison of resonant and nonresonant network reconstruction methods with stronger observational noise. The original network is the same as in Figure \ref{fig:no_noise50} but the integrated time series were perturbed by i.i.d.\ random variables uniformly distributed on $[-4\varepsilon,4\varepsilon]$ (mean $0$, variance $16\varepsilon^{2}/3$) before the reconstruction step. In this case, the all-pairs reconstruction is quite different from the original network, while the resonant reconstruction still captures the resonant subnetwork in a robust way.}}
   \label{fig:strong_noise50}
 \end{figure} 
\subsubsection*{Second order resonant reconstruction}
The second reconstruction method that we present can be used when \eqref{equationsofmotionintro} does not contain any first order resonant terms, that is, when $A_{j,k}^{(1)}=0$ if $\langle \omega, k\rangle =0$. In this case, the first order resonant terms do not provide any information on the network topology. We will nevertheless show that 
\begin{align}\label{secondtorderapproxresonant}
\nonumber 
 \phi_j(T) \!-\! \phi_j(0) \!-\! \omega_j T = \varepsilon^2 T \!\! \sum_{\langle \omega,k\rangle = 0} \!\! C_{j,k}^{(2)} e^{i\langle k, \phi(0)\rangle}  + \mathcal{O}( \varepsilon)\, \ \mbox{when}\  T \sim \varepsilon^{-\frac{3}{2}}\, ,
 \end{align}
where the $C_{j,k}^{(2)}$ are the Fourier coefficients of the so-called {\it second order resonant normal form} of \eqref{equationsofmotionintro}. It follows that, up to a small correction of the order $\sim \varepsilon$, over timescales of the order $T\sim\varepsilon^{-\frac{3}{2}}$, solutions to \eqref{equationsofmotionintro} are once more governed by a slow linear drift of the order $\varepsilon^2 T \sim \varepsilon^{\frac{1}{2}}$. This drift is again determined only by resonant terms, while the impact of nonresonant terms is too small to be distinguished from observational noise. Similarly to the first order method described above, our second order resonant reconstruction method fits a resonant time-$T$ map to observations of phase drifts over a timescale $T\sim\varepsilon^{-\frac{3}{2}}$,  by computing least squares estimators $\hat C_{j,k}$ of the actual coefficients  $C_{j,k}^{(2)}$ of the second order normal form. We are also able to formulate verifiable conditions under which the reconstructed coefficients provide a reliable approximation of the true second order normal form coefficients, see Theorem \ref{thm:2ndordermainthm}, and we illustrate the method with a numerical example.

\subsubsection*{Discussion}
 A key challenge in data-driven model reconstruction is that the observations of the underlying dynamical system are often limited or noisy. This may lead to an ill-conditioned regression problem, overfitting, and non-uniqueness in the inferred dynamical equations. It is then common to select a reconstruction using statistical regularization, such as LASSO or Basis Pursuit. For example, the Sparse Identification of Nonlinear Dynamics (SINDy) method  \cite{brunton_discovering_2016} in principle allows one to fit a large and possibly redundant function library (e.g., consisting of Fourier basis functions) to the observations, but selects a {\it sparse} reconstruction by penalizing for the number of nonzero basis functions that this reconstruction contains (e.g., the number of nonzero Fourier coefficients.) This method 
has been used successfully for reconstructing biological  \cite{mangan_inferring_2016} and chemical \cite{nijholt_emergent_2022} oscillator networks, and to explain observations of brain dynamics \cite{delabays_hypergraph_2025}.  The Algorithm for Revealing Network Interactions (ARNI) \cite{casadiego_model-free_2017}
 works in a similar manner. A related algorithm to stably reconstruct sparse networks from limited data was recently developed in   \cite{novaes_recovering_2021, pereira_robust_2025}.

It was observed in \cite{nijholt_emergent_2022} that
SINDy sometimes automatically selects a resonant normal form as reconstruction of a coupled oscillator network. The authors of \cite{nijholt_emergent_2022} explain this by showing that
the normal form is often the sparsest ODE that generates the same dynamics as the true network up to a certain error margin. In this paper, we adopt a contrasting approach, by imposing in advance that our reconstruction consists only of resonant terms. We thereby guarantee the sparsity of the reconstruction a priori, which allows us to avoid statistical regularization altogether.
  
An obvious disadvantage of resonant reconstruction is that only some of the properties of a phase oscillator network are encoded in its 
 normal form. These properties may not include all the strengths of all the couplings between the nodes in the network. For example, the first order normal form contains only the resonant first order interactions. The relation between a phase oscillator network and its second order normal form is much more intricate, as we show in Appendix \ref{app:2ndorderNF}. 
 
 \subsubsection*{Organization of the paper}

The remainder of this paper is organized as follows. In Section \ref{sec:framework} we show why the assumption of bounded noise is natural in the context of reconstructing phase reductions. In Section \ref{sec:1method} we present our first order reconstruction method. We prove the validity of this method in Sections \ref{sec:approxdriftdyn} and \ref{sec:validitysection} and investigate its performance through an example in Section \ref{sec:examplefirstorder}.  In Section \ref{sec:2ndordersection} we present the second order reconstruction method. Its validity is proved in Section  \ref{sec:validity2norder}, and an application of the method is presented in Section \ref{sec:2ndorderexample}. Several technical results on resonant normal forms are included in the appendices.

\section{Phase reductions and bounded noise} \label{sec:framework}
Before presenting and analyzing our resonant reconstruction method in detail, we will first show in this section that small bounded uncertainties in the phase signals arise naturally when attempting to reconstruct phase equations for high-dimensional nonlinear oscillator systems. 
To see how phase equations are typically derived in this context, 
let us consider a network of weakly coupled differential equations of the form
\begin{equation} \label{eq:highdimoscillators}
\dot x_j = f_j(x_j) + \varepsilon g_j(x_1, \ldots, x_n) \ \mbox{with} \ x_j\in \mathbb{R}^{d_j}\ \mbox{and} \ 1\leq j \leq n\, ,
\end{equation}
in which each ODE $\dot x_j=f_j(x_j)$ in isolation possesses a stable limit cycle $t\mapsto \gamma_j(t)$ of minimal period $T_j$ and frequency $\omega_j:=2\pi/T_j$. In the uncoupled situation (i.e., when $\varepsilon=0$), the product of these limit cycles defines an $n$-dimensional normally hyperbolic invariant torus $\mathbf{T}^{(0)}$ in the phase space $\mathbb{R}^N$ (where $N:=d_1+\ldots+d_n$) of \eqref{eq:highdimoscillators}. One can assign to each point $x\in \mathbf{T}^{(0)}$ a unique vector of phases 
$$\phi = (\phi_1, \ldots, \phi_n) \in \mathbb{T}^n = (\mathbb{R}/2\pi\mathbb{Z})^n$$ 
defined by the equality 
$$x=E^{(0)}(\phi) := (\gamma_1(\phi_1/\omega_1), \ldots, \gamma_n(\phi_n/\omega_n)) \in \mathbf{T}^{(0)}\, ,$$
 in terms of which the dynamics on the torus
 takes the  form 
 $$\dot \phi_j = \omega_j\ \mbox{for} \ 1\leq j \leq n\,  .$$ 
 The flow on $\mathbf{T}^{(0)}$ is thus quasi-periodic. 
F\'enichel's theorem \cite{fenichel_persistence_1971} states that \eqref{eq:highdimoscillators} will admit an invariant torus $\mathbf{T}^{(\varepsilon)}$  close to $\mathbf{T}^{(0)}$ for small $\varepsilon\neq 0$. 
Any parametrization $E^{(\varepsilon)}: \mathbb{T}^n \to \mathbb{R}^N$ of $\mathbf{T}^{(\varepsilon)}$ of the form 
\begin{equation}
\label{eq:Edef}
E^{(\varepsilon)}(\phi) = E^{(0)}(\phi) 
+ \varepsilon E^{(1)}(\phi) +  \varepsilon^2 E^{(2)}(\phi) + \ldots
\end{equation}
 again associates a unique phase vector $\phi\in \mathbb{T}^n$ to each point $x\in \mathbf{T}^{(\varepsilon)}$, and the ODEs governing their evolution are now of the form  
\begin{align}\label{equationsofmotion}
\dot \phi_j = \omega_j + \varepsilon F_j^{(1)}(\phi)  +   \varepsilon^{2} F_j^{(2)}(\phi) + \ldots \,  .
\end{align} 
Equation \eqref{equationsofmotion} is  called 
a {\it phase reduction} of \eqref{eq:highdimoscillators} and it describes the dynamics on $\mathbf{T}^{(\varepsilon)}$ in terms of the phase variables defined by \eqref{eq:Edef}. Phase reductions can be computed analytically from the original equations \eqref{eq:highdimoscillators} by a  variety of asymptotic methods, see for instance \cite{brown_phase_2004,nakao_phase_2016,rosenblum_numerical_2019,mau_high-order_2023,ashwin_dead_2021,leon_analytical_2023,nicks_insights_2024,bick_higher-order_2024,leon_theory_2025}. 

Because the persisting torus $\mathbf{T}^{(\varepsilon)}$ is  (locally or globally) attracting, the solutions to \eqref{eq:highdimoscillators} that are actually observed in numerical simulations or experiments will typically lie on or close to  $\mathbf{T}^{(\varepsilon)}$. As a result,  data-driven reconstruction methods for coupled oscillator systems often aim to recover the phase equations \eqref{equationsofmotion} instead of the full equations \eqref{eq:highdimoscillators}. This is usually done by taking a collection of (possibly noisy) finite time observations $X^{(m)}(t)\in \mathbb{R}^N$ (with $1\leq m \leq M$) of true solutions $x^{(m)}(t)$ to \eqref{eq:highdimoscillators} lying on or close to $\mathbf{T}^{(\varepsilon)}$. 
Each signal $X^{(m)}(t)$ is then transformed into a phase signal, i.e., a function $\Phi^{(m)}(t) \in \mathbb{R}^n$  that  satisfies
\begin{equation}
\label{eq:phaseevolution}
\Phi^{(m)}_j(T) - \Phi^{(m)}_j(0)  = \omega_j T + \mathcal{O}(\varepsilon T)\   \mbox{for}\ 1\leq j \leq n.
\end{equation}
Such phase signals can be constructed in various ways, for instance by applying delayed embedding or a Hilbert transform to the $X^{(m)}(t)$, often combined with a choice of filtering technique, see \cite{matsuki_network_2024, rosenblum_identification_2002, kralemann_reconstructing_2011, kralemann_reconstructing_2014}. The main idea is that the $\Phi^{(m)}(t)$ will be approximations of lifts of true solutions $\phi^{(m)}(t)$ of a phase reduction  \eqref{equationsofmotion} of \eqref{eq:highdimoscillators}.  
The final step of the reconstruction method  consists of fitting a nonlinear phase vector field  to the phase signals $\Phi^{(m)}(t)$ as discussed in Section \ref{sec:intro}. 

Limited, noisy or ill-conditioned observations $X^{(m)}(t)$ may have a negative effect on the quality of the reconstruction \eqref{fittedvf} of \eqref{equationsofmotion}. 
Moreover, even with perfect knowledge of the phase equations \eqref{equationsofmotion}, it can be challenging to infer information about the original high-dimensional equations \eqref{eq:highdimoscillators}. One well-known problem is that the phase equations need not reflect the exact network topology of \eqref{eq:highdimoscillators}. Specifically, it was found in~\cite{nijholt_emergent_2022, kralemann_reconstructing_2011, kralemann_reconstructing_2014, stankovski_coupling_2015, rosenblum_inferring_2023} that the phase equations \eqref{fittedvf} and \eqref{equationsofmotion} 
may display multi-body interactions, even when the original equations \eqref{eq:highdimoscillators} are governed by two-body (``dyadic'') interactions only. In~\cite{kralemann_reconstructing_2011, kralemann_reconstructing_2014, rosenblum_inferring_2023}, these ``virtual'' or ``emergent'' multi-body interactions have been explained as a result of the phase reduction process, while in~\cite{nijholt_emergent_2022} they were interpreted as resonant terms in the normal form of the phase equations. 

This paper is motivated by a more fundamental problem that complicates the interpretation of reconstructed phase equations, and which to our knowledge has not received  attention in the literature on network reconstruction: the fact that 
phase reductions are not unique, because the choice of phase variables on the invariant torus $\mathbf{T}^{(\varepsilon)}$ is not unique, see \cite{gracht_parametrisation_2023}.
 This nonuniqueness means that the extracted phase signals $\Phi^{(m)}(t)$ may not all be (approximations of) solutions of the same dynamical system. 
 
 To explain this, let us fix a choice of ``reference'' phase variables $\phi\in \mathbb{T}^n$, defined by a reference embedding $x=E^{(\varepsilon)}(\phi) \in \mathbf{T}^{(\varepsilon)}$. The time-evolution of these phase variables is then governed by an ODE of the form \eqref{equationsofmotion}.  
 In view of equation \eqref{eq:phaseevolution}, we can also assume that the observations $\Phi^{(m)}(t)$ approximate the evolution of some phase variables, say 
 \begin{align}\label{eq:1stLeq}
 |\psi^{(m)}_j(t) - \Phi^{(m)}_j(t)| \leq L_j^{(1)} \varepsilon\ \mbox{for}\ 1\leq j\leq n\ \mbox{and}\ 1\leq m \leq M\, ,
  \end{align}
where the phases $\psi\in \mathbb{T}^n$ are defined through some embedding $x=\widetilde{E}^{(\varepsilon)}(\psi)$. However, it is not necessarily the case that $\widetilde E^{(\varepsilon)} = E^{(\varepsilon)}$, and the embedding $\widetilde E^{(\varepsilon)}$ may even be different for each $1\leq m \leq M$. 
 The $\psi^{(m)}(t)$ may therefore not be solutions to the reference equations \eqref{transformedphaseequations}. 
Indeed, when $E^{(\varepsilon)}(\phi) = E^{(0)}(\phi) \!+\! \varepsilon E^{(1)}(\phi) \!+\! \varepsilon^2 \ldots \! =\!  E^{(0)}(\psi) \!+\! \varepsilon \widetilde E^{(1)}(\psi) \!+\! \varepsilon^2 \ldots  = \widetilde E^{(\varepsilon)}(\psi)$ are two different embeddings of $\mathbf{T}^{(\varepsilon)}$, then 
  \begin{equation}\label{trafointro}
\psi = H^{(\varepsilon)}(\phi) = \phi + \varepsilon H^{(1)}(\phi)+ \varepsilon^2\ldots 
  \end{equation}
  for some smooth change of coordinates $H^{(\varepsilon)}: \mathbb{T}^n \to \mathbb{T}^n$. Using that the equations of motion for the phases $\phi$ are given by the interaction functions $F^{(1)}, F^{(2)}, \ldots$ as in \eqref{equationsofmotion},  it follows that
  \begin{align}\label{transformedphaseequations}
  \dot \psi_j   = \omega_j + \varepsilon \left( F^{(1)}_j(\psi) + \sum_{k=1}^n  \omega_k \left(\partial_{k} H^{(1)}_j\right)(\psi)  \right)+ \varepsilon^2\ldots  \, . 
  \end{align}
 This confirms that the evolution of the phases $\psi=(\psi_1, \ldots \psi_n)$ is governed by interaction functions $\widetilde F^{(1)}, \widetilde F^{(2)}, \ldots$ that are generally different from those governing the evolution of the phases $\phi=(\phi_1, \ldots, \phi_n)$. On the other hand, equation \eqref{trafointro} implies that
 \begin{align}
 \label{eq:2ndLeq}
 |\psi^{(m)}_j(t) -\phi^{(m)}_j(t) | \leq L_j^{(2)}\varepsilon\ \mbox{for}\ 1\leq j\leq n\ \mbox{and}\ 1\leq m \leq M\, ,
 \end{align}
 for some uniform constant $L_j^{(2)}>0$. Together \eqref{eq:1stLeq} and \eqref{eq:2ndLeq}  imply that $$|\phi^{(m)}_j(t)-\Phi^{(m)}_j(t)| \leq |\phi^{(m)}_j(t)-\psi^{(m)}_j(t)| + |\psi^{(m)}_j(t)-\Phi^{(m)}_j(t)|  \leq L_j \varepsilon\, , $$
 where $L_j:=L_j^{(1)}+L_j^{(2)}$.  
 The term $|\psi^{(m)}_j(t)-\Phi^{(m)}_j(t)|$ in this expression can be thought of as observational noise (which may even vanish when the observations are perfect.) The term $|\phi^{(m)}_j(t)-\psi^{(m)}_j(t)|$ reflects the intrinsic mismatch between reference phases and observed phases. 
 
In practice, it may be hard to distinguish these two sources of uncertainty in the phase signals. We  therefore do not make this distinction in this paper,  by simply assuming that $|\phi^{(m)}_j(t)-\Phi^{(m)}_j(t)| \leq L_j \varepsilon$ for a uniform constant $L_j$.   
 
\section{First order resonant reconstruction}\label{sec:1method}  

We now present the details of our first order resonant reconstruction method for equations of the form \eqref{equationsofmotionintro}. Our method relies on the assumption that we are given reasonable estimates of the coupling strength $\varepsilon$ and the oscillator frequencies $\omega_1, \ldots, \omega_n$ in \eqref{equationsofmotionintro} from the start. 

\begin{remark}
For a fixed value of $\varepsilon\neq 0$, the constant Fourier-term (i.e., the sum of the $\phi$-independent terms)  in equation \eqref{equationsofmotionintro} is 
$$\Omega_j = \omega_j + \varepsilon A_{j,0}^{(1)} + \varepsilon^2A_{j,0}^{(2)} + \mathcal{O}(\varepsilon^3)\, .$$
This means that the frequency $\omega_j$ is only defined up to a term of the order $\sim \varepsilon$. It may in fact be easier to estimate $\Omega_j$ from observations than $\omega_j$. We shall thus assume that we know the value of $\omega_j$ up to an error of the order $\sim\varepsilon$. We  simply include the unknown  part of $\omega_j$ as unknown constant terms $\varepsilon A_{j,0}^{(1)}, \varepsilon^2 A_{j,0}^{(2)}, \ldots$ in  \eqref{equationsofmotionintro}. Our reconstruction method can estimate these terms. 
Similarly, the Fourier-coefficient of the term $e^{i\langle k, \phi\rangle}$ in \eqref{equationsofmotionintro} is
$$\alpha_{j,k} := \varepsilon A_{j,k}^{(1)} + \varepsilon^2A_{j,k}^{(2)} +\mathcal{O}(\varepsilon^3)\, .$$
This means the values of the coupling strength  $\varepsilon$ and the coupling constants $A_{j,k}^{(1)}, A_{j,k}^{(2)}, \ldots$ are not unique:  only their orders of magnitude  are well-defined.  It suffices to fix a value of $\varepsilon$ such that  $A_{j,k}^{(1)}, A_{j,k}^{(2)}, \ldots \sim 1$. 
\end{remark}
\noindent Our method for reconstructing equations \eqref{equationsofmotionintro} furthermore assumes that we are given noisy phase signals 
$$\Phi^{(m)}(t)\in\mathbb{R}^n \ \mbox{with}\ 1\leq m \leq M\, .$$
In fact, we only need to know the value of these phase signals  at times $$t=0 \ \mbox{and} \ t=T:=\frac{1}{\sqrt{\varepsilon}}\, .$$
As discussed above, we assume the noise in these phase signals to be small and bounded. More precisely, we will assume  that there are constants $L_j>0$ such that for every $1\leq m \leq M$, there is a (lift of an) exact solution $\phi^{(m)}(t)\in \mathbb{R}^n$ to  \eqref{equationsofmotionintro} 
so that 
\begin{equation}\label{eq:approxphi}
| \phi^{(m)}_j(0)-\Phi^{(m)}_j(0)| \leq L_j\varepsilon\ \mbox{and} \ | \phi^{(m)}_j(T)-\Phi^{(m)}_j(T)| \leq L_j\varepsilon\, .
\end{equation}
Our first order reconstruction method can now be summarized as follows:
\begin{enumerate}
    \item[{\bf 1)}] Fix an oscillator $1\leq j \leq n$ and calculate the rescaled observed drifts
    \[
    {\Delta}^{(m)} :=  \frac{\Phi^{(m)}_j(T) - \Phi^{(m)}_j(0) -  \omega_j T}{\varepsilon T} \in \mathbb{R} \ \ \mbox{for}\  1\leq m \leq M\, .
    \]
    These drifts together form a (column) vector $\Delta\in \mathbb{R}^M$.  

    \item[{\bf 2)}] Select a  finite library of {\it resonant} Fourier labels 
    $$ \mathcal{K} = \{k_1, \ldots, k_K\}    \subset \{k\in \mathbb{Z}^n\,| \, \langle k,  \omega\rangle = 0\} \, ,$$ 
    where $K:=|\mathcal{K}|$, and evaluate the functions $e^{i \langle k, \phi \rangle}$  ($k\in \mathcal{K}$) at the observed initial states $\Phi^{(m)}(0)$, to create the library matrix
    $$ \Theta := \left( \begin{array}{ccc}
   e^{i\langle k_1, \Phi^{(1)}(0) \rangle} & \hdots & e^{i\langle k_K, \Phi^{(1)}(0)\rangle}  \\
    \vdots & \ddots & \vdots \\
    e^{i\langle k_1, \Phi^{(M)}(0)\rangle} & \hdots & e^{i\langle k_K, \Phi^{(M)}(0)\rangle} 
 \end{array} \right)  : \mathbb{C}^{K} \to \mathbb{C}^M\, .$$
We assume that the observations $\Phi^{(m)}(0)$ (with $m=1,\ldots, M$) are sufficiently well spread over $\mathbb{T}^n$  that $\Theta$ is injective. 
    \item[{\bf 3)}] List the resonant Fourier coefficients $A_{j,k}^{(1)}$ (with $k\in \mathcal{K}$) in the  vector $$A_j^{(1)} : = (A^{(1)}_{j, k_1}, \ldots, A^{(1)}_{j, k_K})^T \in\mathbb{C}^K\, .$$ 
    Inspired by \eqref{firstorderapproxresonant}, we estimate this vector by the vector
    $$\hat A_j = (\hat A_{j, k_1}, \ldots, \hat A_{j, k_K})^T \in\mathbb{C}^K$$ 
    defined by  
    $$\hat A_j = {\arg\min}_{B\in \mathbb{C}^K} \sum_{m=1}^M  \left|   \Delta^{(m)} - \sum_{i=1}^{K}    B_{k_i} e^{i\langle k_i, \Phi^{(m)}(0)\rangle}   \right|^2 \in \mathbb{C}^K \, . 
    $$
     It is well-known that this minimizer is unique when $\Theta$ is injective, and given by
    \begin{equation}\label{pseudo-inverseformula}
     \hat A_j = \Theta^+ (\Delta) \in \mathbb{C}^K\, ,
      \end{equation}
    where $\Theta^+ := (\Theta^H\Theta)^{-1} \Theta^H$ is the Moore-Penrose pseudo-inverse of $\Theta$ and $\Theta^H : = \overline{\Theta}^T$ its Hermitian transpose. We prove \eqref{pseudo-inverseformula} in Appendix \ref{app:pseudoinverse} for completeness.
\end{enumerate}
This reconstruction method was used to generate Figures  \ref{fig:no_noise50}, \ref{fig:weak_noise50} and \ref{fig:strong_noise50} in the introduction. In the upcoming  sections we will rigorously prove that the $\hat A_{j,k}$ are good estimators of the $A_{j,k}^{(1)}$, when certain natural conditions are satisfied. We give an application of the method in Section \ref{sec:examplefirstorder}  

\section{Approximating the drift dynamics}
\label{sec:approxdriftdyn}
 The proof that the $\hat A_{j,k}$ are close to the true Fourier coefficients $A_{j,k}^{(1)}$ hinges on Lemma \ref{lem:approxlemma} which implies that, to leading order, the flow of  \eqref{equationsofmotionintro} is governed by first order resonant terms only.  To show this, we first prove Lemma \ref{simplelemma} by exploiting a combination of Picard iteration and Taylor expansion. 
By $C^{k}(\mathbb{T}^n)$ we denote the space of real-valued $C^k$-functions on $\mathbb{T}^n$ and for an $f\in C^0(\mathbb{T}^n)$ we shall write 
$$\|f\|_{0} := \sup_{\phi\in\mathbb{T}^n} |f(\phi)|\, .$$ 
For a  continuous function $g: \mathbb{T}^n\times [-\varepsilon_0, \varepsilon_0]\to \mathbb{R}$ we  similarly denote $\|g\|_0:=\sup_{   (\phi,\varepsilon)\in\mathbb{T}^n\times[-\varepsilon_0, \varepsilon_0] } |g(\phi, \varepsilon)|$.

\begin{lemma} \label{simplelemma}
   For all $1\leq j\leq n$, assume that  $F^{(1)}_j \in C^{1}(\mathbb{T}^n)$ and $R^{(2)}_j \in C^0(\mathbb{T}^n\times[-\varepsilon_0, \varepsilon_0])$. Let  $|\varepsilon| \leq \varepsilon_0$ and let \(\phi(t)\) be a solution to  
  \begin{equation}\label{eq:eqmotion}
  \dot \phi_j = \omega_j + \varepsilon F_j^{(1)}(\phi) + \varepsilon^2 R_j^{(2)}(\phi, \varepsilon)\ \mbox{for all}\ 1 \leq j \leq n\, .
  \end{equation}
   Then 
\begin{equation}\label{eq:picardformula1}
\phi_j(T) - \phi_j(0) - \omega_jT = \varepsilon\int_0^T F^{(1)}_j(\phi(0) + \omega t) \, dt + r_j^{(2)}(T, \varepsilon)  \, , 
\end{equation}
in which   
$$| r_j^{(2)}(T,\varepsilon) | \leq   A_j \varepsilon^2  T^2 + B_j \varepsilon^2 T\, ,$$
  for  constants $A_j, B_j$ independent of $T$ and $\varepsilon$, given in \eqref{eq:Aj} and \eqref{eq:Bj}.
\end{lemma}
\begin{proof}
Let us write 
$R_i^{(1)}(\phi, \varepsilon) := F_i^{(1)}(\phi) + \varepsilon R_i^{(2)}(\phi, \varepsilon)$. Then $R_i^{(1)}$ is a bounded continuous function on $\mathbb{T}^n\times [-\varepsilon_0, \varepsilon_0]$, and 
$\dot \phi_i(s) = \omega_i + \varepsilon R_i^{(1)}(\phi(s),\varepsilon)$ for all $1\leq i\leq n$.   
Integrating this equation from $s=0$ to $s=t$ yields 
\begin{equation}\label{eq:firstintegral}
\phi_i(t) = \phi_i(0)+ \omega_i t + \varepsilon \int_0^t R^{(1)}_i(\phi(s), \varepsilon)\, ds \, .
\end{equation}
We see in particular that 
\begin{equation}\label{eq:firstestimate}
    \phi_i(t) = \phi_i(0) + \omega_i t + r_i^{(1)}(t, \varepsilon)
\ \ \mbox{in which}\ \ 
|r_i^{(1)}(t, \varepsilon)| \leq |\varepsilon t| \| R_i^{(1)}\|_0  \, .
\end{equation}
Next, we integrate the equation $\dot \phi_j(t)=\omega_j + \varepsilon F_j^{(1)}(\phi(t)) + \varepsilon^2R_j^{(2)}(\phi(t),\varepsilon)$ from $t=0$ to $t=T$ to find that
\begin{align} \label{eq:secondintegral} 
& \phi_j(T)  = \phi_j(0) + \omega_j T + \varepsilon \int_0^T F^{(1)}_j\left(\phi(t)\right) dt + \varepsilon^2\int_0^T \! R^{(2)}_j(\phi(t),\varepsilon) \, dt .
\end{align}
To rewrite the first integral term, we use \eqref{eq:firstestimate} to obtain  
\begin{equation}\label{eq:F1formula}
F_j^{(1)}\left(\phi(t)\right) =  F^{(1)}_j\left(\phi(0)+\omega t +  r^{(1)}(t, \varepsilon ) \right) = F^{(1)}_j\left(\phi(0)+\omega t \right) +   h_j(t, \varepsilon) 
\end{equation}
in which, in view of \eqref{eq:firstestimate},   
\begin{equation} \label{eq:rjbar}
|h_j(t, \varepsilon)|\leq \sum_{i=1}^n \| \partial_iF^{(1)}_j \|_0  |r_i^{(1)}(t, \varepsilon) | \leq |\varepsilon t| \sum_{i=1}^n  \| \partial_iF^{(1)}_j \|_0  \| R^{(1)}_i\|_0 \, .
\end{equation}
Integration of \eqref{eq:F1formula} from $t=0$ to $t=T$ therefore gives
\begin{align} \label{eq:firstintegralestimate} 
& \left| \varepsilon \! \int_0^T \!\!F^{(1)}_j\left(\phi(t)\right)\, dt  - \varepsilon \! \int_0^T\!\! F^{(1)}_j\left(\phi(0)+\omega t\right) \, dt  \right| \leq   
   \varepsilon \int_0^T \!\! | h_j(t, \varepsilon)| \, dt \leq A_j \varepsilon^2T^2
   \end{align}
   for the constant
\begin{align}\label{eq:Aj}
 A_j:=\frac{1}{2} \sum_{i=1}^n  \| \partial_iF^{(1)}_j \|_0  \| R^{(1)}_i\|_0\, .
\end{align}
The second term in \eqref{eq:secondintegral} can clearly be bounded by
\begin{equation}\label{eq:secondintegralestimate}
\left| \varepsilon^2\int_0^T R^{(2)}_j(\phi(t),\varepsilon) \, dt \right|  \leq B_j \varepsilon^2 T
\end{equation}
for
\begin{equation}\label{eq:Bj}
  B_j :=\|R_j^{(2)}\|_0\, .
\end{equation}
Together, \eqref{eq:secondintegral}, \eqref{eq:firstintegralestimate} and \eqref{eq:secondintegralestimate} prove the Lemma.
\end{proof}
\noindent We now investigate why and in which exact sense Lemma \ref{simplelemma} implies formula \eqref{firstorderapproxresonant}. Recall from \eqref{eq:picardformula1} that the drift $\phi_j(T)-\phi_j(0)-\omega_jT$ of a solution to \eqref{equationsofmotionintro} can be approximated by the integral 
\begin{align}\label{eq:intoverflow}
\varepsilon\int_0^T F_j^{(1)}(\phi(0)+\omega t)\, dt
\end{align}
of $\varepsilon F_j^{(1)}$ over the unperturbed flow $t\mapsto \phi(0)+\omega t$. One can expect that this integral is of the order $\varepsilon T$, and in particular that it is of the order $\sqrt{\varepsilon}$ when $T$ is of the order $\frac{1}{\sqrt{\varepsilon}}$. On the other hand, Lemma \ref{simplelemma} implies that the difference $r_j^{(2)}(T,\varepsilon)$ between the true drift and \eqref{eq:intoverflow} is of the order $\varepsilon$ for $T\sim \frac{1}{\sqrt{\varepsilon}}$. This means that, to leading order in $\varepsilon$ and over timescales of the order $\frac{1}{\sqrt{\varepsilon}}$, the drift is determined by the integral \eqref{eq:intoverflow}. 

We now analyze \eqref{eq:intoverflow} in detail by  splitting $F^{(1)}_j$ into a {\it library} part, a {\it remaining resonant} part and a {\it nonresonant} part, by writing 
$$F_j^{(1)}(\phi) = F_j^{\mathcal{K}}(\phi) + F_j^{\mathcal{R}}(\phi) + F_j^{\mathcal{N}}(\phi)\, ,$$
where 
\begin{align} \label{eq:FK}
F_j^{\mathcal{K}}(\phi) & = \sum_{k \in \mathcal{K}} A_{j,k}^{(1)} e^{i\langle k, \phi \rangle }\, , \\ \label{eq:FR}
F_j^{\mathcal{R}}(\phi) & = \sum_{\langle k, \omega\rangle = 0,\, k \notin \mathcal{K}} A_{j,k}^{(1)} e^{i\langle k, \phi \rangle }\, , \\ 
\label{eq:FN}
F_j^{\mathcal{N}}(\phi) & = \sum_{\langle k, \omega\rangle \neq 0} A_{j,k}^{(1)} e^{i\langle k, \phi \rangle }\, .
\end{align}
Here, $\mathcal{K}  \subset \left\{ k \in \mathbb{Z}^n \mid \langle \omega, k \rangle = 0 \right\}$ is a finite   library of resonant Fourier labels. We integrate these functions separately. 
\\ \mbox{}\\ 
{\bf Library part \eqref{eq:FK}:} Because $\mathcal{K}$ is finite and $\langle  k, \phi(0)+\omega t\rangle = \langle k, \phi(0)\rangle$ when $k$ is resonant, it is clear that 
\begin{align}\label{FKexpansion} 
\varepsilon\! \int_0^T\!\! F_j^{\mathcal{K}}(\phi(0)+\omega t) \, dt = \varepsilon \sum_{k\in \mathcal{K}} A_{j,k}^{(1)} \int_0^T \!\! e^{i\langle k, \phi(0)+\omega t\rangle} \, dt = \varepsilon T \sum_{k\in \mathcal{K}} A_{j,k}^{(1)}  e^{i\langle k, \phi(0)\rangle} \, . 
\end{align}
Thus, the value of this integral increases slowly and linearly in $T$, as was anticipated in formula \eqref{firstorderapproxresonant}. Note that over timescales of the order $T\sim \frac{1}{\sqrt{\varepsilon}}$, the right hand side of \eqref{FKexpansion} will typically be of the order $\sqrt{\varepsilon}$. 
\\ \mbox{}\\
\noindent {\bf Remaining resonant part \eqref{eq:FR}:}
Note that $F_j^{\mathcal{R}}$ may potentially have an infinite Fourier series. We can nevertheless formally compute the integral  of \eqref{eq:FR} term-by-term. This yields
\begin{align}\label{Fremainingresonantintegral} 
\varepsilon \int_0^T F_j^{\mathcal{R}}(\phi(0)+\omega t)\, dt & \sim \varepsilon T \sum_{\langle k, \omega\rangle = 0,\, k\notin \mathcal{K}} A_{j,k}^{(1)} e^{i\langle k, \phi(0) \rangle} \, .
\end{align}
We would like this integral to converge and to be at most of the order $\varepsilon$ on a timescale of the order $T\sim \frac{1}{\sqrt{\varepsilon}}$.  This can be ensured by requiring, for example, that 
\begin{align}\label{eq:remainingresonanttermsestimate}
\sum_{\langle k, \omega \rangle = 0, \, k\notin \mathcal{K}} |A_{j,k}^{(1)}| \leq C_j \sqrt{\varepsilon}\, ,
\end{align}
meaning that the library is large enough to describe the most relevant resonant terms in $F_j^{(1)}$. 
Under condition \eqref{eq:remainingresonanttermsestimate}, equation \eqref{Fremainingresonantintegral} becomes an equality, and we have 
\begin{align}\label{Fremainingresonantintegralbound} 
\left| \varepsilon \int_0^T F_j^{\mathcal{R}}(\phi(0)+\omega t)\, dt \right|  & \leq C_j \varepsilon^{\frac{3}{2}}T\, .
\end{align}
Note that the right hand side of \eqref{Fremainingresonantintegralbound} is of the order $\varepsilon$ when $T\sim\frac{1}{\sqrt{\varepsilon}}$. It is thus of the order of the error term $r_j^{(2)}(T, \varepsilon)$ in \eqref{eq:picardformula1} and considerably smaller than the right hand side of \eqref{FKexpansion}.
\\ \mbox{} \\
{\bf Nonresonant part \eqref{eq:FN}:}
The nonresonant part $F_j^{\mathcal{N}}$ may similarly have an infinite Fourier series, so that  its integral can only formally be computed to be
\begin{align} \label{Fnonresonantintegral}
\varepsilon \int_0^T & F_j^{\mathcal{N}}(\phi(0)+\omega t)\, dt \sim 
\varepsilon \sum_{\langle k, \omega\rangle \neq 0}  A_{j,k}^{(1)} \int_0^T e^{i\langle k, \phi(0)+\omega t \rangle}\, dt \notag \\ 
& = \varepsilon  \sum_{\langle k, \omega\rangle \neq 0}  \frac{A_{j,k}^{(1)}  e^{i\langle k, \phi(0)\rangle}}{i\langle k, \omega \rangle}   \left( e^{i\langle k,   \omega\rangle T} - 1\right)\, . 
\end{align}
Each term in this formal sum is an oscillatory and uniformly bounded function of $T$. We would like the sum to be bounded itself and of the order $\varepsilon$. This can be ensured by asking, for example, that 
\begin{align}\label{eq:nonresonantestimate}
    \sum_{\langle k, \omega \rangle \neq 0} \left| \frac{A_{j,k}^{(1)}}{\langle k, \omega\rangle}\right|  \leq \frac{1}{2} D_j \, . 
\end{align}
Under condition \eqref{eq:nonresonantestimate}, equation \eqref{Fnonresonantintegral} becomes an equality, and
\begin{align} \label{Fnonresonantintegralbound}
\left| \varepsilon \int_0^T  F_j^{\mathcal{N}}(\phi(0)+\omega t)\, dt \right|  \leq D_j \varepsilon \, .
\end{align}
Being of the order $\varepsilon$ (uniformly in time), this integral is of the order of the error term $r_j^{(2)}$ in \eqref{eq:picardformula1} and considerably smaller than the right hand side of \eqref{FKexpansion} on the timescale $T\sim \frac{1}{\sqrt{\varepsilon}}$. 
\\ \mbox{}\\
\noindent The following result follows almost immediately from the discussion above:
\begin{lemma} \label{lem:approxlemma}
    Let $0 \leq \varepsilon \leq \varepsilon_0$ be given, and let $\mathcal{K}\subset \{k\in \mathbb{Z}^n\, | \, \langle k,\omega\rangle = 0\, \}$ be a finite set of resonant Fourier labels. Let us decompose $$F^{(1)}_j= F_j^{\mathcal{K}} + F_j^{\mathcal{R}} + F_j^{\mathcal{N}}\, , $$
    as defined in \eqref{eq:FK}, \eqref{eq:FR} and \eqref{eq:FN}. 
    Let the assumptions of Lemma \ref{simplelemma} be satisfied and  assume in addition that for all $\phi\in \mathbb{T}^n$, 
    $$\left| \int_0^T F_j^{\mathcal{R}}(\phi+\omega t)\, dt \right|  \leq C_j \sqrt{\varepsilon} T \ \mbox{and}\  \left| \int_0^T  F_j^{\mathcal{N}}(\phi+\omega t)\, dt \right|  \leq D_j $$
    for some constants $C_j, D_j$ independent of $T$ and $\varepsilon$. 
    Then any solution $t\mapsto \phi(t)$ to \eqref{eq:eqmotion} satisfies
$$\phi_j(T)-\phi_j(0)-\omega_j T = \varepsilon T \sum_{k\in \mathcal{K}} A_{j,k}^{(1)}  e^{i\langle k, \phi(0)\rangle} + \rho_j^{(2)}(T,\varepsilon) \, , $$
in which 
$$\left| \rho_j^{(2)}(T,\varepsilon) \right| \leq A_j \varepsilon^2 T^2 + B_j\varepsilon^2 T + C_j \varepsilon^{\frac{3}{2}}T + D_j\varepsilon\, ,$$
and where the constants $A_j,B_j$ are given in \eqref{eq:Aj} and \eqref{eq:Bj}.
\end{lemma}
\begin{proof}
The proof was essentially given in the discussion above. Note that Lemma \ref{simplelemma} states that 
$$\phi_j(T)-\phi_j(0) - \omega_j T = \varepsilon\int_0^T F_j^{(1)}(\phi(0) + \omega t)\, dt + r_j^{(2)}(T,\varepsilon)$$
with  $|r_j^{(2)}(T,\varepsilon)|\leq A_j\varepsilon^2T^2 + B_j \varepsilon^2 T$. The assumptions on $F_j^{\mathcal{R}}$ and $F_j^{\mathcal{N}}$ imply that \eqref{Fremainingresonantintegralbound} and \eqref{Fnonresonantintegralbound} hold. It follows from this that 
$$\left| \varepsilon\int_0^T F_j^{(1)}(\phi(0) + \omega t)\, dt - \varepsilon T \sum_{k\in \mathcal{K}} A_{j,k}^{(1)}  e^{i\langle k, \phi(0)\rangle}  \right| \leq C_j \varepsilon^{\frac{3}{2}}T  + D_j \varepsilon\, . $$
The statement of the lemma now follows.
\end{proof}
\noindent Lemma \ref{lem:approxlemma} implies in particular that $|\rho_j^{(2)}(T,\varepsilon)| \sim \varepsilon$ 
for $T\sim \frac{1}{\sqrt{\varepsilon}}$. Because $\varepsilon T \sim \sqrt{\varepsilon} \gg \varepsilon$ when $T\sim \frac{1}{\sqrt{\varepsilon}}$, this means that  
$$\phi_j(T)-\phi_j(0)-\omega_jT = \varepsilon T \sum_{k\in \mathcal{K}} A_{j,k}^{(1)}  e^{i\langle k, \phi(0)\rangle} + \mathcal{O}(\varepsilon)\ \mbox{for} \ T \sim \frac{1}{\sqrt{\varepsilon}}\, .$$  
Because $\varepsilon T \sim \sqrt{\varepsilon} \gg \varepsilon$ for $T\sim \frac{1}{\sqrt{\varepsilon}}$, this implies that the drifts are determined, to leading order, solely by resonant library terms. 

\section{Validity of first order reconstruction}\label{sec:validitysection}
Exploiting Lemma \ref{lem:approxlemma} of the previous section, we will now prove that our first order resonant reconstruction method produces good estimates $\hat A_{j,k}$ of the true resonant Fourier coefficients $A_{j,k}^{(1)}$ as soon as the observed initial data $\Phi^{(m)}(0)$ (for $m=1, \ldots, M$) are
sufficiently well-distributed over $\mathbb{T}^n$. This condition on the data will be made precise in Lemma \ref{deltalemma} below. We first formulate Theorem \ref{general_thm}, which estimates the error in the reconstructed coefficients  in terms of operator norms of the pseudo-inverse of  $\Theta$. To state the theorem, we define
for a vector $x\in \mathbb{C}^r$, 
 $$ \|x\|_1:= \sum_{i=1}^r |x_i| \ , \ 
\|x\|_{2} := \left(\sum_{i=1}^r |x_i|^2\right)^{\frac{1}{2}} \ \mbox{and} \ 
\|x\|_{\infty} := \max_{1 \leq i \leq r} |x_i|\, .  
$$ 
For a linear operator $Y: \mathbb{C}^q  \to \mathbb{C}^r$ we denote its operator norms by $\| Y\|_{*} = \sup_{\|x\|_*=1} \|Y(x)\|_*$, for $*\in\{1,2,\infty\}$.  It is well-known that these operator norms are given by
\begin{align} \label{eq:1opnormdef}
& \|Y\|_{1} := \sup_{\| x\|_1=1} \|Y(x)\|_1 = \max_{1 \leq j \leq r} \sum_{i=1}^q |Y_{ij}|\, , \\  \label{eq:2opnormdef}
& \|Y\|_{2} := \sup_{\| x\|_2=1} \|Y(x)\|_2 =  \sqrt{ \lambda_{\rm max}(Y^{H}Y)} \ , \\  \label{eq:infopnormdef}
& \|Y\|_{\infty} := \sup_{\| x\|_{\infty}=1} \|Y(x)\|_{\infty}  =  \max_{1 \leq i \leq q} \sum_{j=1}^r |Y_{ij}|\, .
\end{align}
In formula \eqref{eq:2opnormdef}, $\lambda_{\rm max}(Y^{H}Y) = \lambda_{\rm max}(YY^{H})$ is the maximal eigenvalue of $Y^{H}Y$, that is, $\|Y\|_2$ is equal to the largest singular value of $Y$.  Also, note that $\|Y\|_1=\|Y^{H}\|_{\infty}$. 
Obviously, all operator norms are sub-multiplicative, that is, $\|XY\|_* \leq \|X\|_*\|Y\|_*$. 
We can now formulate the main result of this section:

\begin{theorem}\label{general_thm}
    Let the assumptions of Lemmas \ref{simplelemma} and \ref{lem:approxlemma} be satisfied. Write $\mathcal{K} = \{k_1, \ldots, k_K\}$, with $K=|\mathcal{K}|$, and respectively denote by  
    $$A_{j}^{(1)} = (A_{j,k_1}^{(1)}, \ldots, A_{j,k_K}^{(1)})^{T} \in \mathbb{C}^K\ \mbox{and}\  \hat A_{j} = (\hat A_{j,k_1}, \ldots, \hat A_{j,k_K})^{T} \in \mathbb{C}^K$$     the vector of true resonant Fourier coefficients $A_{j,k}^{(1)}$ (with $k\in \mathcal{K}$) in $F_j^{(1)}$, and the vector of their estimates $\hat A_{j,k}$ provided by our first order resonant reconstruction method. Let $\mathfrak{C}_j$ be the constant given in \eqref{eq:Cjdef} below, which is independent of $\varepsilon, K$ and $M$.
Then
\begin{enumerate}
    \item The $l_1$-distance between $\hat A_j$ and $A_j^{(1)}$ satisfies 
    $$\|\hat A_j - A_j^{(1)}\|_1 \leq  M \mathfrak{C}_j \|\Theta^+\|_1 \sqrt{\varepsilon}\, ; $$ 
    \item The $l_2$-distance between $\hat A_j$ and $A_j^{(1)}$ satisfies
    $$\|\hat A_j - A_j^{(1)}\|_2 \leq  \sqrt{M}\mathfrak{C}_j \|\Theta^+\|_2 \sqrt{\varepsilon}\, ;$$ 
    \item The $l_{\infty}$-distance between $\hat A_j$ and $A_j^{(1)}$ satisfies
    $$\|\hat A_j - A_j^{(1)}\|_{\infty} \leq  \mathfrak{C}_j \|\Theta^+\|_{\infty} \sqrt{\varepsilon}\, .$$ 
\end{enumerate}
\end{theorem}

\begin{proof}
Recall the assumption in \eqref{eq:approxphi} that $\Phi^{(m)}(0)$ and $\Phi^{(m)}(T)$ are observations of true solutions $t\mapsto \phi^{(m)}(t)$ to \eqref{equationsofmotionintro}  that satisfy
$$| \phi^{(m)}_j(0)-\Phi^{(m)}_j(0)| \leq L_j\varepsilon\ \mbox{and} \ | \phi^{(m)}_j(T)-\Phi^{(m)}_j(T)| \leq L_j\varepsilon\, .$$ 
Using these true solutions, we can define the true rescaled drifts 
$$\delta^{(m)} := \frac{\phi_j^{(m)}(T) - \phi_j^{(m)}(0) - \omega_j T}{\varepsilon T}\, .$$
The observed rescaled drifts are good approximations of these true rescaled drifts, as it follows immediately from \eqref{eq:approxphi} that  
\begin{equation}
    \label{eq:deltaDeltaestimate}
| \Delta^{(m)} - \delta^{(m)}| \leq \frac{2 L_j}{T}\, , 
\end{equation}
which is small for large $T$. 
 Dividing equation \eqref{eq:picardformula1} in Lemma \ref{simplelemma} by $\varepsilon T$, we find that the true rescaled  drifts satisfy
 \begin{equation}
    \label{repeatDeltaequation}
 \delta^{(m)} = \frac{1}{T} \int_0^T F_j^{(1)}(\phi^{(m)}(0) + \omega t)\, dt +  e^{(m)}\, ,
 \end{equation}
for certain approximation errors that satisfy $|e^{(m)}| \leq A_j\varepsilon T + B_j \varepsilon$.  Combining \eqref{eq:deltaDeltaestimate}  and \eqref{repeatDeltaequation}  yields that 
\begin{equation}
    \label{repeatrepeatDeltaequation}
 \Delta^{(m)} = \frac{1}{T} \int_0^T F_j^{(1)}(\phi^{(m)}(0) + \omega t)\, dt +  \hat e^{(m)}\, ,
 \end{equation}
with $|\hat e^{(m)}| \leq A_j\varepsilon T + B_j \varepsilon + 2L_jT^{-1}$. 

Using that $F_j^{(1)}$ is continuously differentiable  by assumption, we have 
\begin{align}\notag
& |F_j^{(1)}(\phi^{(m)}(0)+\omega t)   - F_j^{(1)}(\Phi^{(m)}(0) +\omega t) | \leq  \\ \notag  &  \sum_{i=1}^n\|\partial_iF_j^{(1)}\|_0 |\phi^{(m)}_i(0) - \Phi^{(m)}_i(0)| 
 \leq \varepsilon \left( \sum_{i=1}^n\|\partial_iF_j^{(1)}\|_0 L_i\right) \, ,
\end{align}
in view of \eqref{eq:approxphi}. As a result, we can now rewrite \eqref{repeatrepeatDeltaequation} as
\begin{equation}\label{eq:DeltaintegralPhi}
\Delta^{(m)} = \frac{1}{T}\int_0^T F_j^{(1)}(\Phi^{(m)}(0) + \omega t)\, dt +  \tilde e^{(m)} \, ,
\end{equation}
in which 
\begin{equation}\label{eq:rhoestimate}
|\tilde e^{(m)}| \leq \varepsilon \left( \sum_{i=1}^n\|\partial_iF_j^{(1)}\|_0 L_i\right) + |\hat e^{(m)} | \leq A_j\varepsilon T +  \widetilde B_j \varepsilon + 2L_jT^{-1} \, ,
\end{equation}
 for the new constant 
 \begin{equation}
     \label{eq:Bjtilde}
 \widetilde B_j := B_j + \left( \sum_{i=1}^n\|\partial_iF_j^{(1)}\|_0 L_i\right)\, . 
 \end{equation}
Next, recall that we split $F_j^{(1)}=F_j^{\mathcal{K}} + F_j^{\mathcal{R}} + F_j^{\mathcal{N}}$. Under the assumptions of Lemma \ref{lem:approxlemma}, the main contribution to the integral in \eqref{eq:DeltaintegralPhi} is 
$$ \frac{1}{T} \int_0^T F_j^{{\mathcal{K}}}(\Phi^{(m)}(0) + \omega t)\, dt =  \sum_{k\in \mathcal{K}} A_{j,k}^{(1)} e^{i\langle k, \Phi^{(m)}(0)\rangle}\, ,$$
 because the contributions of the remaining resonant part and the nonresonant part can be estimated by
 \begin{align} \notag 
 & \left| \frac{1}{T} \int_0^T F_j^{\mathcal{R}}(\Phi^{(m)}(0) + \omega t)\, dt \right| \leq C_j \sqrt{\varepsilon} \, ,
 \\ \notag 
 & \left| \frac{1}{T} \int_0^T F_j^{\mathcal{N}}(\Phi^{(m)}(0) + \omega t)\, dt \right|   \leq D_jT^{-1}\, ,
 \end{align}
 with $C_j$ and $D_j$ as given in Lemma \ref{lem:approxlemma}.   Equation \eqref{eq:DeltaintegralPhi} thus reduces to
\begin{align}  \label{eq:deltamequation}    
\Delta^{(m)}  =  \sum_{k\in \mathcal{K}} A_{j,k}^{(1)} e^{i\langle k, \Phi^{(m)}(0)\rangle} +  \rho^{(m)} \ \mbox{for all} \ 1\leq m \leq M\, .
\end{align}
with 
\begin{equation}\label{eq:rhomestimate}
| \rho^{(m)}| \leq \mathfrak{c}_j:= A_j \varepsilon T + \widetilde B_j\varepsilon + C_j \sqrt{\varepsilon} + \widetilde D_jT^{-1}\, , 
\end{equation}
for the new constant 
\begin{equation}\label{eq:Djtilde}
\widetilde D_j:=D_j+2L_j \, . 
\end{equation}
Note that the bound $\mathfrak{c}_j$ defined in \eqref{eq:rhomestimate} is of the order $\sqrt{\varepsilon}$ when $T \sim \frac{1}{\sqrt{\varepsilon}}$. More specifically, 
\begin{equation}\label{eq:Cjdef}
| \rho^{(m)}| \leq \mathfrak{c}_j = \mathfrak{C}_j\sqrt{\varepsilon} \ \mbox{when}\ T=\frac{1}{\sqrt{\varepsilon}}\, ,
\ \mbox{for}\ \mathfrak{C}_j:=A_j+\widetilde B_j + C_j + \widetilde D_j \, .
\end{equation}
Because the upper bound $\mathfrak{c}_j$ in \eqref{eq:rhomestimate} is independent of $m$, it is easy to estimate the norms of vector $\rho=(\rho^{(1)}, \ldots, \rho^{(M)})^T \in\mathbb{R}^M$. Indeed equation \eqref{eq:rhomestimate} and the definitions of the norms imply that
\begin{align}\label{rho1bound}
& \|\rho\|_1  \leq   M \mathfrak{C}_j\sqrt{\varepsilon}  \, , \\ 
\label{rho2bound}
& \|\rho\|_2  \leq   \sqrt{M}\mathfrak{C}_j\sqrt{\varepsilon}  \, , \\
\label{rhoinfbound}
& \|\rho\|_{\infty}  \leq  \mathfrak{C}_j\sqrt{\varepsilon} \, . 
\end{align}
To finish the proof, note that we can write equation \eqref{eq:deltamequation} as the equality of $M$-vectors
\begin{equation}\label{DeltaThetaequation}
\Delta = \Theta(A_j^{(1)}) + \rho\, . 
\end{equation} 
Now we recall that our first order resonant reconstruction method chooses the estimated  Fourier coefficients to be given by $\hat A_j = \Theta^+(\Delta)$, see \eqref{pseudo-inverseformula}. 
Using \eqref{DeltaThetaequation} and the fact that $\Theta^+\Theta = (\Theta^H\Theta)^{-1}\Theta^H\Theta = \rm{Id}_{\mathbb{C}^K}$, it follows that 
$$\hat A_j - A_j^{(1)} = \Theta^+(\Delta) -\Theta^+\Theta(A_j^{(1)}) = \Theta^+(\Delta) - \Theta^+(\Delta) + \Theta ^+(\rho) = \Theta ^+(\rho) \, ,$$
 and in particular that 
 \begin{equation}\label{eq:generalbound}
 \|\hat A_j - A_j^{(1)}\|_* \leq \|\Theta^+\|_* \| \rho\|_* \ \mbox{for} \ * \in \{1,2,\infty\}\, .
 \end{equation}
 The theorem follows by combining \eqref{eq:generalbound} with \eqref{rho1bound}, \eqref{rho2bound} and \eqref{rhoinfbound}.
\end{proof}
\noindent The norms $\|\Theta^+\|_*$ depend on the initial values $\Phi^{(m)}(0)$ of the observations. In (numerical) experiments,  these norms can be computed explicitly, thus making the error bounds for the reconstructed Fourier coefficients provided by Theorem \ref{general_thm} fully explicit.  The numbers $\|\Theta^+\|_*$ thus provide an a posteriori measure of the quality of the data on which the reconstruction is based.
\begin{remark}
To gain intuition on how the operator norms $\|\Theta^+\|_*$ depend on the initial data $\Phi^{(m)}(0)$, let us remark that the pseudo-inverse $\Theta^{+} = (\Theta^H \Theta)^{-1}\Theta^H$ can be written as the product
$$\Theta^{+} = S^{-1} T , \  \mbox{for}\ S := \frac{1}{M}\Theta^{H}\Theta  \ \mbox{and} \ T:= \frac{1}{M} \Theta^{H} \, . $$
The entries of the $K\times K$-matrix $S$ are given by 
$$ S_{p,q} = \frac{1}{M}\sum_{m=1}^{M} e^{i\langle k_q-k_p, \Phi^{(m)}(0) \rangle} \, .$$
If the initial data $\Phi^{(m)}(0)$ (for $m=1, \ldots, M$) were selected randomly from the uniform distribution on $\mathbb{T}^n$ and  $M\to\infty$, then $S$ would almost surely be close to the $K\times K$-identity matrix ${\rm Id}_{\mathbb{C}^K}$ by the law of large numbers. The norms $\|S-{\rm Id}_{\mathbb{C}^K}\|_*$ are thus a measure of how well-distributed the initial conditions are over $\mathbb{T}^n$. Lemma \ref{deltalemma} gives bounds on the operator norms of $\Theta^+$ in terms of these norms.   
\begin{lemma}\label{deltalemma} 
Let $0<\delta < 1$ be a constant. 
    \begin{enumerate}
    \item If
    $
    \| S - {\rm Id}_{\mathbb{C}^K} \|_{1} \leq \delta,
    $
    then
    $
    \|\Theta^{+}\|_{1} \leq \frac{K}{M} \frac{1}{1-\delta}
    $.
     \item If
    $
    \| S - {\rm Id}_{\mathbb{C}^K} \|_{2} \leq \delta,
    $
    then
    $
    \|\Theta^{+}\|_{2} \leq \frac{1}{\sqrt{M}} \frac{\sqrt{1+\delta}}{1-\delta}$.
        \item If
    $
    \| S - {\rm Id}_{\mathbb{C}^K} \|_{\infty} \leq \delta,
    $
    then
    $
    \|\Theta^{+}\|_{\infty} \leq \frac{1}{1 - \delta}
    $.
    \end{enumerate}
\end{lemma}
\begin{proof}
Recall that $\Theta^{+} = S^{-1}T$. The assumption that $\|S-{\rm Id}_{\mathbb{C}^K}\|_* \leq \delta<1$ implies that $S$ is invertible and that $\|S^{-1}\|_*\leq \frac{1}{1-\delta}$. Thus, it remains to find an estimate on the norms of $T: \mathbb{C}^M\to\mathbb{C}^K$. We recall that the matrix coefficients  of this map are given by 
$$T_{p, m}=\frac{1}{M} e^{-i\langle k_p, \Phi^{(m)}(0)\rangle}\, .$$
\begin{enumerate}
    \item[{\it 1.}] It is clear that
$$\|T\|_1 =\max_m \frac{1}{M}\sum_{p=1}^K \left|e^{-i\langle k_p, \Phi^{(m)}(0)\rangle} \right| = \frac{K}{M}.$$ 
Therefore, $\|\Theta^{+}\|_{1} \leq \|S^{-1}\|_{1} \|T\|_{1}   \leq \frac{K}{M} \frac{1}{1-\delta}$. 
     \item[{\it 2.}] 
Because $T=\frac{1}{M}\Theta^{H}$, it holds that $T^{H}T = \frac{1}{M^2}\Theta\Theta^{H}$. Therefore,  
\begin{align}
\|T\|_2^2  & =  \lambda_{\max}(T^{H} T) =  
\frac{1}{M^2} \lambda_{\max}(\Theta \Theta^{H})
\nonumber \\ \nonumber 
& = \frac{1}{M^2} \lambda_{\max}(\Theta^{H} \Theta)= \frac{1}{M} \lambda_{\max}(S) =\frac{1}{M} \|S\|_2.
\end{align}
Thus, $\|\Theta^{+}\|_{2} \leq \|S^{-1}\|_{2} \|T\|_{2} =  \|S^{-1}\|_{2} \sqrt{\|S\|_{2}/M} \leq  \frac{1}{\sqrt{M}} \frac{\sqrt{1+\delta}}{1-\delta}$.
      \item[{\it 3.}] It is clear that
$$\|T\|_\infty =\max_p  \frac{1}{M}\sum_{m=1}^M \left|e^{-i\langle k_p, \Phi^{(m)}(0)\rangle} \right| = 1.$$ 
Therefore, $\|\Theta^{+}\|_{\infty} \leq \|S^{-1}\|_{\infty} \|T\|_{\infty}   \leq \frac{1}{1-\delta}$. 
\end{enumerate}
 
\end{proof}

\end{remark}

\section{Example: the Kuramoto equations}
\label{sec:examplefirstorder}
We now apply our first order reconstruction method to a network of classical Kuramoto oscillators of the form
\begin{equation}
\dot{\phi}_j
=
\omega_j
+ \varepsilon (\Delta\omega)_j + 
\varepsilon \sum_{i=1}^n a_{ij}\sin(\phi_i-\phi_j),
\qquad j=1,\dots,n,
\label{eq:nearres_kuramoto}
\end{equation}
where we allow for a small and bounded ``detuning'' in the oscillator frequencies,  satisfying $|\varepsilon(\Delta\omega)_j| \leq \varepsilon \Omega$ for all $1\leq j\leq n$ and some  constant $\Omega>0$. The connection index $a_{ij}\in\{0,1\}$ indicates the presence or absence of an interaction from $i$ to $j$ in the oscillator network. 

Equations \eqref{eq:nearres_kuramoto} are of the form \eqref{equationsofmotionintro}, for  $R^{(2)}_j(\phi, \varepsilon)=0$, and 
\begin{align}\notag
\ F_j^{(1)}(\phi) = \!\! \sum_{k\in \mathbb{Z}^n} \! A_{j, k}^{(1)} e^{i\langle k, \phi \rangle }\ \mbox{with}\ A_{j,k}^{(1)} = \left\{ \!\!\! \begin{array}{rl} (\Delta\omega)_j & \mbox{when}\ k=0\, , \\ a_{ij}/ 2i & \mbox{when} \ k = e_i-e_j\neq 0 \, ,\\  -a_{ij}/ 2i & \mbox{when} \ k = e_j-e_i\neq 0 \, ,\\ 0 & \mbox{otherwise} \, . \end{array} \right. 
\end{align}
Here $e_j\in\mathbb{Z}^n$ denotes the $j$-th standard basis vector, that is, $(e_{j})_j =1$ and  $(e_{j})_i =0$ for $i\neq j$. Note that $a_{ij}= i(A^{(1)}_{j,e_i-e_j}- A^{(1)}_{j,e_j-e_i})$.

Our goal is to estimate the coefficients $a_{ij}$ from noisy observations of solutions to \eqref{eq:nearres_kuramoto}. However, recall that our first order resonant reconstruction method only recovers the resonant part of $F^{(1)}$, which is given by 
$$
\overline{F}^{(1)}_j(\phi) := \sum_{\langle k, \omega \rangle = 0} A_{j,k}^{(1)}e^{i\langle k,\phi\rangle} = \sum_{\omega_i=\omega_j} a_{ij}\sin(\phi_i-\phi_j)  \, .
$$
The reconstruction method will therefore  be able to estimate exactly those $a_{ij}$ for which $\omega_i=\omega_j$. In order to do this, we choose the $j$-dependent reconstruction library 
\begin{equation}
\mathcal{K}_j
=  
\left\{\pm(e_i-e_j)\in\mathbb{Z}^n \ :\\ 1\leq i\leq n \ \mbox{satisfies}\ \omega_i=\omega_j 
\right\}\, ,
\label{eq:nearres_library}
\end{equation}
to reconstruct the equations of motion for $\phi_j$. Note that $0\in \mathcal{K}_j$. The resonant reconstruction method will then yield estimates $$\widehat{(\Delta\omega)}_j:= \hat A_{j,0} \ \mbox{of}\ (\Delta\omega)_j \ \mbox{and} \ \hat a_{ij} : = i(\hat A_{j, e_i-e_j} - \hat A_{j, e_j-e_i}) \ \mbox{of} \ a_{i,j}\, .$$ 

\subsection*{Numerical results}
\label{subsec:kuramoto_network_example}

We now illustrate this reconstruction by  two numerical examples.  
Concretely, we consider \eqref{eq:nearres_kuramoto} on a realization of a directed Erd\H{o}s--R\'enyi network with $n=10$ vertices and edge probability $p=0.2$. This network is shown in Figure \ref{fig:kuramoto_network} and has 
the adjacency matrix 
\begin{equation}
\label{eq:kuramoto_example_adjacency}
a=
\begin{pmatrix}
*&0&0&{\bf 0}&0&0&{\bf 0}&1&{\bf 0}&{\bf 0}\\
1&*&1&0&{\bf 0}&{\bf 1}&0&0&0&0\\
0&1&*&0&1&0&0& {\bf 0}&0&1\\
{\bf 0}&0&0&*&0&1&{\bf 0}&0&{\bf 0}&{\bf 0}\\
0&{\bf 0}&0&1&*&{\bf 0}&1&0&0&0\\
0&{\bf 0}&0&0&{\bf 0}&*&0&1&1&0\\
{\bf 1}&0&0&{\bf 1}&0&0&*&0&{\bf 0}&{\bf 0}\\
0&0&{\bf 0}&0&0&0&0&*&0&0\\
{\bf 1}&0&0&{\bf 1}&0&0&{\bf 0}&0&*&{\bf 0}\\
{\bf 1}&0&0&{\bf 0}&1&0&{\bf 0}&0&{\bf 0}&*
\end{pmatrix}\, .
\end{equation}
In the first numerical experiment we consider equations \eqref{eq:nearres_kuramoto} with coupling strength $\varepsilon=\frac{1}{25}$ and detuned oscillator frequencies 
\begin{align*}\label{detunedFreq}
\omega+\varepsilon\Delta\omega
= 
& (1.9960,\,3.0002,\,1.0035,\,1.9932,\,2.9910, \\
& \,2.9968,\,1.9922,\,0.9936,\,2.0077,\,1.9973)\,  .
\end{align*}
Note that these detuned frequencies deviate from the ``tuned frequencies'' 
\begin{equation} \label{baseFreq}
\omega=(2,3,1,2,3,3,2,1,2,2)
\end{equation}
by an amount $|\varepsilon (\Delta\omega)_j|\leq \frac{\varepsilon}{4}$. In other words, we have that $\Omega\leq \frac14$. The resonant  
 $a_{ij}$ (those for which $\omega_i=\omega_j$) are written in boldface in \eqref{eq:kuramoto_example_adjacency}.
\begin{figure}[ht]
    \centering
    \includegraphics[width=.55\textwidth]{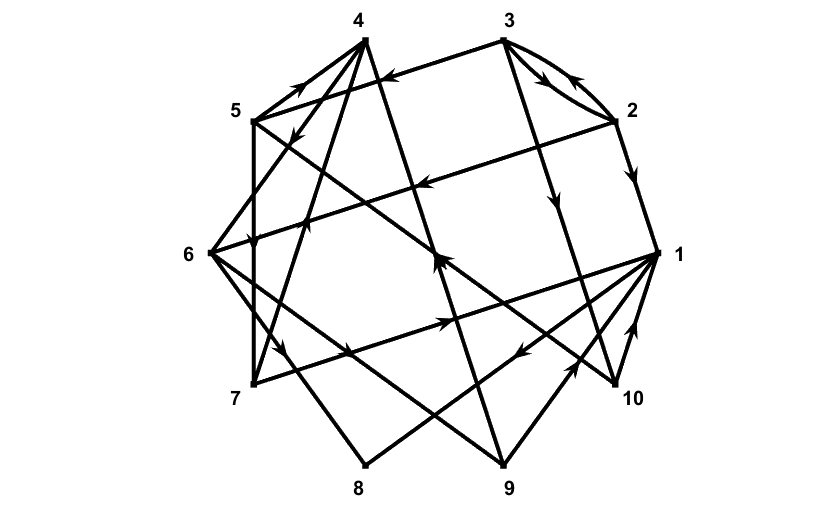}
    \caption{Directed network with adjacency matrix given in \eqref{eq:kuramoto_example_adjacency}.}
    \label{fig:kuramoto_network}
\end{figure}

\noindent To test the resonant reconstruction method, we chose $M=25$ initial conditions $\phi^{(m)}(0)$ randomly from the uniform distribution on $\mathbb{T}^{10}$. Starting from these initial conditions, we  integrated equations  \eqref{eq:nearres_kuramoto} in MATLAB over a time interval $T=\frac{1}{\sqrt{\varepsilon}} = 5$ to obtain the final states $\phi^{(m)}(T)$. Bounded uniform noise of mean zero and magnitude at most $\varepsilon$ was added to  the initial and final states, to produce the noisy observations $\Phi^{(m)}(0)$, $ \Phi^{(m)}(T)$. By design, $|\Phi^{(m)}_j(0)-\phi^{(m)}(0)| \leq \varepsilon$ and $|\Phi^{(m)}_j(T)-\phi^{(m)}(T)| \leq \varepsilon$, i.e., $L_j=1$ for all $j$. We then calculated the noisy rescaled drifts $\Delta^{(m)}= \frac{\Phi^{(m)}_j(T) - \Phi^{(m)}_j(0) - \omega_jT}{\varepsilon T}$ for each oscillator, 
 and computed the reconstructed resonant interaction coefficients $\hat a_{ij} = i(\hat A_{j, e_i-e_j} - \hat A_{j, e_j-e_i})$ using the procedure explained in Section \ref{sec:1method}. We present the resulting reconstructed  coefficients in the left column of Table \ref{tab:coef_table_reduced}.
 Note that each reconstructed coefficient satisfies $|\hat a_{ij}- a_{ij}| \leq 0.1881 < \sqrt{\varepsilon}=\frac{1}{5}$. We thus obtain a  reliable reconstruction of the ``resonant subnetwork'' of the network that generated the  noisy time series.
 
 We  repeated the numerical experiment for the  same network, with the exact same frequencies, but with a smaller coupling strength $\varepsilon=\frac{1}{100}$. Now  $M=100$ initial conditions were chosen uniformly at random, from which equations \eqref{eq:nearres_kuramoto} were integrated over a time $T=\frac{1}{\sqrt{\varepsilon}}=10$. The resulting reconstructed coefficients are presented in the right column of Table \ref{tab:coef_table_reduced}.
 They all satisfy $|\hat a_{ij}- a_{ij}| \leq 0.0118 \ll \sqrt{\varepsilon}=\frac{1}{10}$.
 
\begin{table}[H]
\centering
{\small
\begin{tabular}{cc}
\begin{tabular}{|l|r|r|r|}
\hline
\textbf{Coeff.} & \textbf{Value} & \textbf{True} & \textbf{Error} \\
\hline
$\hat a_{4,1}$   & 0.0625  & 0 & 0.0625 \\
\hline
$\hat a_{7,1}$   & 1.0357  & 1 & 0.0357 \\
\hline
$\hat a_{9,1}$   & 1.0734  & 1 & 0.0734 \\
\hline
$\hat a_{10,1}$  & 1.0164  & 1 & 0.0164 \\
\hline
$\hat a_{5,2}$   & 0.0274  & 0 & 0.0274 \\
\hline
$\hat a_{6,2}$   & -0.0037 & 0 & 0.0037 \\
\hline
$\hat a_{8,3}$   & -0.0870 & 0 & 0.0870 \\
\hline
$\hat a_{1,4}$   & 0.0939  & 0 & 0.0939 \\
\hline
$\hat a_{7,4}$   & 0.9474  & 1 & 0.0526 \\
\hline
$\hat a_{9,4}$   & 0.8119  & 1 & 0.1881 \\
\hline
$\hat a_{10,4}$  & 0.1049  & 0 & 0.1049 \\
\hline
$\hat a_{2,5}$   & -0.0108 & 0 & 0.0108 \\
\hline
$\hat a_{6,5}$   & 0.0846  & 0 & 0.0846 \\
\hline
$\hat a_{2,6}$   & 0.9185  & 1 & 0.0815 \\
\hline
$\hat a_{5,6}$   & 0.0332  & 0 & 0.0332 \\
\hline
$\hat a_{1,7}$   & 0.0915  & 0 & 0.0915 \\
\hline
$\hat a_{4,7}$   & 0.0614  & 0 & 0.0614 \\
\hline
$\hat a_{9,7}$   & -0.0182 & 0 & 0.0182 \\
\hline
$\hat a_{10,7}$  & 0.0717  & 0 & 0.0717 \\
\hline
$\hat a_{3,8}$   & 0.0910  & 0 & 0.0910 \\
\hline
$\hat a_{1,9}$   & 0.0292  & 0 & 0.0292 \\
\hline
$\hat a_{4,9}$   & -0.0708 & 0 & 0.0708 \\
\hline
$\hat a_{7,9}$   & 0.0582  & 0 & 0.0582 \\
\hline
$\hat a_{10,9}$  & 0.0057  & 0 & 0.0057 \\
\hline
$\hat a_{1,10}$  & -0.1123 & 0 & 0.1123 \\
\hline
$\hat a_{4,10}$  & 0.0275  & 0 & 0.0275 \\
\hline
$\hat a_{7,10}$  & 0.0131  & 0 & 0.0131 \\
\hline
$\hat a_{9,10}$  & 0.0624  & 0 & 0.0624 \\
\hline
\end{tabular}
&
\begin{tabular}{|l|r|r|r|}
\hline
\textbf{Coeff.} & \textbf{Value} & \textbf{True} & \textbf{Error} \\
\hline
$\hat a_{4,1}$   & 0.0107 & 0 & 0.0107 \\
\hline
$\hat a_{7,1}$   & 1.0075 & 1 & 0.0075 \\
\hline
$\hat a_{9,1}$   & 0.9893 & 1 & 0.0107 \\
\hline
$\hat a_{10,1}$  & 0.9947 & 1 & 0.0053 \\
\hline
$\hat a_{5,2}$   & 0.0032 & 0 & 0.0032 \\
\hline
$\hat a_{6,2}$   & 0.0016 & 0 & 0.0016 \\
\hline
$\hat a_{8,3}$   & -0.0037 & 0 & 0.0037 \\
\hline
$\hat a_{1,4}$   & 0.0036 & 0 & 0.0036 \\
\hline
$\hat a_{7,4}$   & 0.9917 & 1 & 0.0083 \\
\hline
$\hat a_{9,4}$   & 1.0068 & 1 & 0.0068 \\
\hline
$\hat a_{10,4}$  & -0.0079 & 0 & 0.0079 \\
\hline
$\hat a_{2,5}$   & -0.0024 & 0 & 0.0024 \\
\hline
$\hat a_{6,5}$   & 0.0001 & 0 & 0.0001 \\
\hline
$\hat a_{2,6}$   & 1.0028 & 1 & 0.0028 \\
\hline
$\hat a_{5,6}$   & 0.0003 & 0 & 0.0003 \\
\hline
$\hat a_{1,7}$   & 0.0020 & 0 & 0.0020 \\
\hline
$\hat a_{4,7}$   & -0.0079 & 0 & 0.0079 \\
\hline
$\hat a_{9,7}$   & 0.0034 & 0 & 0.0034 \\
\hline
$\hat a_{10,7}$  & 0.0061 & 0 & 0.0061 \\
\hline
$\hat a_{3,8}$   & 0.0103 & 0 & 0.0103 \\
\hline
$\hat a_{1,9}$   & -0.0099 & 0 & 0.0099 \\
\hline
$\hat a_{4,9}$   & 0.0030 & 0 & 0.0030 \\
\hline
$\hat a_{7,9}$   & -0.0071 & 0 & 0.0071 \\
\hline
$\hat a_{10,9}$  & -0.0087 & 0 & 0.0087 \\
\hline
$\hat a_{1,10}$  & 0.0018 & 0 & 0.0018 \\
\hline
$\hat a_{4,10}$  & -0.0037 & 0 & 0.0037 \\
\hline
$\hat a_{7,10}$  & -0.0070 & 0 & 0.0070 \\
\hline
$\hat a_{9,10}$  & 0.0118 & 0 & 0.0118 \\
\hline
\end{tabular}
\end{tabular}
}
\caption{Reconstructed connection indices for the classical Kuramoto network in Figure \ref{fig:kuramoto_network} with coupling strength $\varepsilon=\frac{1}{25}$ (left) and $\varepsilon=\frac{1}{100}$ (right), respectively based on $M=25$ noisy observations (left) and  $M=100$ noisy observations (right). The true indices are included for comparison.}
\label{tab:coef_table_reduced}
\end{table}

\noindent 
To give an impression of the quality of the numerically obtained time series data, we show the operator norms of the pseudo-inverse $\Theta^+$ in Table \ref{tab:pseudoinverse_norms}.
\begin{table}[H]
\centering
{\small
\begin{tabular}{cc}
\begin{tabular}{|c|c|c|c|}
\hline
\textbf{Node} & \textbf{$\|\Theta^+\|_1$} & \textbf{$\|\Theta^+\|_2$} & \textbf{$\|\Theta^+\|_\infty$} \\
\hline
1  & 0.11666  & 0.13422  & 1.0447 \\
\hline
2  & 0.060023 & 0.11109  & 1.0078 \\
\hline
3  & 0.034191 & 0.10677  & 1.0025 \\
\hline
4  & 0.13197  & 0.14438  & 1.0434 \\
\hline
5  & 0.057883 & 0.10861  & 1.0035 \\
\hline
6  & 0.060877 & 0.11266  & 1.0087 \\
\hline
7  & 0.12615  & 0.1265   & 1.0237 \\
\hline
8  & 0.034191 & 0.10677  & 1.0025 \\
\hline
9  & 0.12067  & 0.13222  & 1.0281 \\
\hline
10 & 0.11993  & 0.13194  & 1.0306 \\
\hline
\end{tabular}
&
\begin{tabular}{|c|c|c|c|}
\hline
\textbf{Node} & \textbf{$\|\Theta^+\|_1$} & \textbf{$\|\Theta^+\|_2$} & \textbf{$\|\Theta^+\|_\infty$} \\
\hline
1  & 0.84499 & 0.6075  & 1.5011 \\
\hline
2  & 0.35756 & 0.27967 & 1.0257 \\
\hline
3  & 0.29472 & 0.32319 & 1.0422 \\
\hline
4  & 0.73352 & 0.49288 & 1.3396 \\
\hline
5  & 0.38461 & 0.28213 & 1.0256 \\
\hline
6  & 0.3397  & 0.27719 & 1.0293 \\
\hline
7  & 0.72467 & 0.4719  & 1.1864 \\
\hline
8  & 0.29472 & 0.32319 & 1.0422 \\
\hline
9  & 0.76565 & 0.46992 & 1.2023 \\
\hline
10 & 0.7613  & 0.40911 & 1.1522 \\
\hline
\end{tabular}
\end{tabular}
}
\caption{Operator norms of $\Theta^+$ for each node.}
\label{tab:pseudoinverse_norms}
\end{table}

\noindent The original Kuramoto network, its resonant subnetwork, and its two reconstructions are visualized in Figure  \ref{fig:rec_m100}.

\begin{figure}[H]
    \centering
    \includegraphics[width=0.95\textwidth]{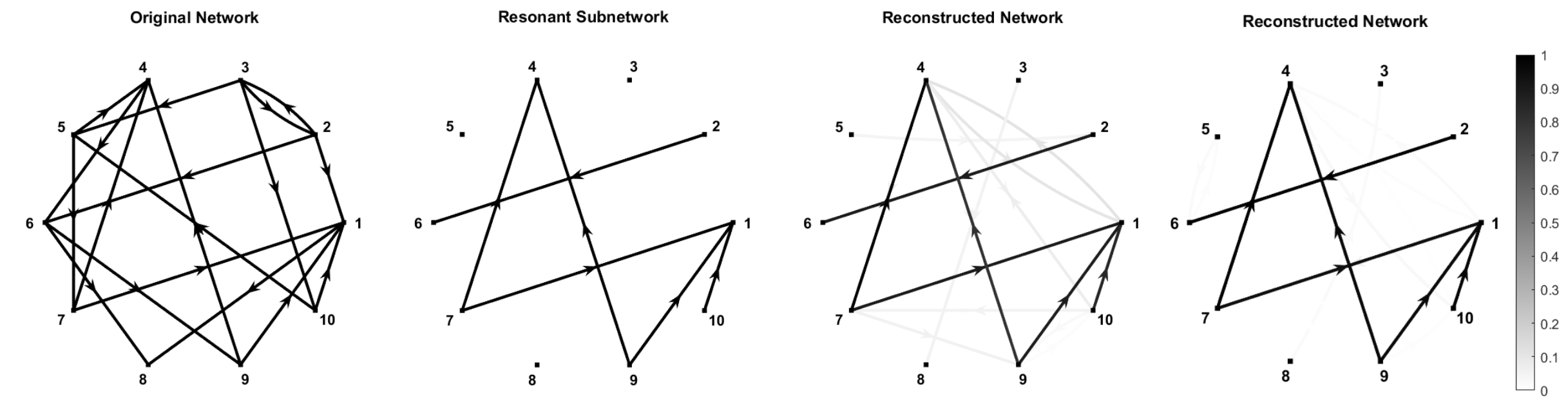}
    \caption{Reconstruction of a Kuramoto network from numerical time series data. From left to right: original network that generated the time series; resonant subnetwork; reconstruction of the resonant subnetwork for $\varepsilon=\frac{1}{25}$ from $M=25$ observations;  reconstruction of the resonant subnetwork for $\varepsilon=\frac{1}{100}$ from $M=100$ observations.}
    \label{fig:rec_m100}
\end{figure}

\subsection*{An explicit error bound}
For equations \eqref{eq:nearres_kuramoto} we now derive an explicit theoretical bound on the reconstruction errors $|\hat a_{ij}-a_{ij}|$. Recall that 
 Theorem \ref{general_thm} gives a bound on the errors $\|\hat A_j-A_j^{(1)}\|_{\infty}$ in the reconstructed complex Fourier coefficients. These in turn lead to error bounds for the estimated $\hat a_{i,j}$, because  
\begin{align}\notag 
& | \hat a_{ij} - a_{ij} | =  |\hat A_{j,e_i-e_j} - \hat A_{j, e_j-e_i} - A^{(1)}_{j,e_i-e_j}+  A^{(1)}_{j,e_j-e_i} |    \\ \notag
& \leq |\hat A_{j,e_i-e_j} - A^{(1)}_{j,e_i-e_j} | + |\hat A_{j,e_j-e_i} - A^{(1)}_{j,e_j-e_i} | \\ \notag 
& \leq 2 \| \hat A_j - A_j^{(1)}\|_{\infty} \leq 2\mathfrak{C}_j \|\Theta^+\|_{\infty} \sqrt{\varepsilon} \, ,   
\end{align}
where the final inequality follows from  part {\it 3} of Theorem \ref{general_thm}. The following proposition provides an explicit upper bound for the constant $\mathfrak{C}_j$.  For simplicity, we assume  that the pure (non-detuned) frequencies satisfy
\begin{equation}
    \label{omegainequality}
\omega_i\neq \omega_j \implies |\omega_i-\omega_j| \geq 1\, .
\end{equation}
\begin{proposition}\label{CjpropKuramoto}
 Consider the classical Kuramoto system \eqref{eq:nearres_kuramoto}. Denote by $d_j:=\sum_{i=1}^n a_{ij}$  the in-degree of cell $j$, and by $d_{\rm max}:=\max_j d_j$ the maximal in-degree. Then the constant $\mathfrak{C}_j$ in Theorem \ref{general_thm} is bounded by 
 \begin{equation}\label{eq:mathfrakCjKuramoto} 
  \mathfrak{C}_j \leq \mathfrak{C}_j^{\rm max}:= \frac{1}{2}d_{j}(d_{\rm max}+4+2L+\Omega) + 2L \, .
  \end{equation}
\end{proposition}
    \begin{proof}
    Recall that the constant $\mathfrak{C}_j$ appearing in Theorem \ref{general_thm} is given by 
    $\mathfrak{C}_j=A_j+\widetilde B_j + C_j + \widetilde D_j$ -- see \eqref{eq:Cjdef}. So it suffices to recall the definitions of these constants, and to compute them for the  example of the Kuramoto equations \eqref{eq:nearres_kuramoto}. 
To do this, note that \eqref{eq:nearres_kuramoto} is of the form \eqref{eq:eqmotion} with $R_j^{(2)}=0$, so that  $R_j^{(1)}:= F_j^{(1)} + \varepsilon R_j^{(2)} = F_j^{(1)}$. This implies that $B_j = 0$, see \eqref{eq:Bj}. It is also clear from \eqref{eq:nearres_kuramoto} that $\|R_j^{(1)}\|_0 = \|F_j^{(1)}\|_0 \leq \Omega + \sum_{i}a_{ij} = \Omega + d_j$, and that $\| \partial_iF_j^{(1)}\|_0 = a_{ij}$ for $i\neq j$. It   thus follows from \eqref{eq:Aj} that 
$$A_j\leq \frac{1}{2}  \sum_{i}a_{ij}(\Omega + \sum_k a_{ki})  \leq  \frac{1}{2}d_j (\Omega +  d_{\rm max}) ,$$
and from \eqref{eq:Bjtilde} that
$$\widetilde B_j \leq \sum_i a_{ij} L =  d_j L  \, .$$  
Next, we split $F_j^{(1)} = F_j^{\mathcal{K}} + F_j^{\mathcal{R}} + F_j^{\mathcal{N}}$ in its library-part, its remaining resonant part, and its nonresonant part. Because the library $\mathcal{K}$ was chosen to contain all resonant terms in $F^{(1)}_j$, we have that $F_j^{\mathcal{R}}=0$, and thus that $C_j=0$, see \eqref{Fremainingresonantintegralbound}. Furthermore, 
$$F_j^{\mathcal{N}}(\phi) = \sum_{\omega_i\neq \omega_j} a_{ij}\sin(\phi_i-\phi_j) = \sum_{\langle k, \omega \rangle \neq 0} A_{j,k}^{(1)}e^{i\langle k,\phi\rangle} \, .$$
Using \eqref{omegainequality}, it follows that
$$\sum_{\langle k, \omega \rangle \neq 0} \left| \frac{A_{j,k}^{(1)}}{\langle k, \omega \rangle} \right| \leq \sum_{\langle k, \omega \rangle \neq 0} \left| A_{j,k}^{(1)} \right| = \sum_{\omega_i\neq \omega_j} a_{ij}   \leq d_j   \, , $$
 so we may choose 
$$D_j= 2d_j  \, ,$$
see \eqref{eq:nonresonantestimate},  and $\widetilde{D}_j = D_j + 2L$, see \eqref{eq:Djtilde}. Together, our estimates for $A_j, \widetilde{B}_j, C_j$ and $\widetilde{D}_j$ yield formula \eqref{eq:mathfrakCjKuramoto}. 
\end{proof}

\section{Second order resonant reconstruction}
\label{sec:2ndordersection}
The first order reconstruction method that we presented in the first part of this paper allows us to detect first order (in $\varepsilon$) resonant coupling terms in phase oscillator systems of the form \eqref{equationsofmotionintro}. For instance, when $\omega_i - \omega_j \sim \varepsilon$, the method can identify terms in the differential equations of the form $\varepsilon a_{ij}\sin(\phi_i-\phi_j)$. By design, the method cannot detect nonresonant interactions. However, the interplay between nonresonant interactions at order $\sim \varepsilon$ can lead to dynamical behavior equivalent to that of a resonant interaction at order $\sim \varepsilon^2$. In the absence of resonant first order interactions, our second order resonant reconstruction method can detect such ``emergent'' second order interactions. In Appendix \ref{app:2ndorderNF},  we explore how these emergent interactions are shaped by nonresonant first order interactions. For instance, we show how nonresonant dyadic interactions at order $\sim \varepsilon$ combine to form resonant triplet interactions at order $\sim \varepsilon^2$. 

The key technical result underlying our second order resonant reconstruction method is the second order resonant normal form theorem. This theorem states that for any coupled oscillator system of the form \eqref{equationsofmotionintro} there exists a change of variables (diffeomorphism)
$$\Gamma^{\varepsilon}_{H^{(1)}}: \phi \mapsto \psi:= \phi + \varepsilon H^{(1)}(\phi)  + \mathcal{O}(\varepsilon^2) \, , $$
defined for all $\phi\in \mathbb{T}^n$ and for small enough values of $\varepsilon$, with the property that $\phi(t)$  satisfies \eqref{equationsofmotionintro} if and only if 
$\psi(t):=\Gamma^{\varepsilon}_{H^{(1)}}(\phi(t))$ satisfies an ODE
\begin{equation}\label{psievolution}
\dot \psi_j = \omega_j + \varepsilon \overline{F}_j^{(1)}(\psi) + \varepsilon^2\overline{F}_j^{(2)}(\psi) + \varepsilon^3\overline{R}_j^{(3)}(\phi, \varepsilon)\, ,
\end{equation}
 such that $\overline{F}^{(1)}$ and $\overline{F}^{(2)}$ are sums of resonant Fourier terms only, i.e., 
 $$\overline{F}_j^{(1)}(\phi) = \sum_{\langle k, \omega\rangle = 0} C_{j,k}^{(1)}e^{i\langle k,\phi\rangle}\ \mbox{and} \  \overline{F}_j^{(2)}(\phi) = \sum_{\langle k, \omega\rangle = 0} C_{j,k}^{(2)}e^{i\langle k,\phi\rangle}\, , $$
 for certain Fourier coefficients $C_{j,k}^{(1)}, C_{j,k}^{(2)}$. We prove this fact in Appendices \ref{app:Lie}, \ref{sec:normalformappendix} and \ref{app:2ndorderNF}. In particular, the solutions of \eqref{equationsofmotionintro} and \eqref{psievolution} differ only by an amount of the order $\sim\varepsilon$.  The Fourier coefficients of $\overline{F}^{(1)}$ and $\overline{F}^{(2)}$ can be expressed in terms of the  Fourier coefficients of the original vector field $F^{(1)}$ and $F^{(2)}$. More precisely, we show in Theorem \ref{first_order_nf_thm} that
  $$C_{j,k}^{(1)}= A_{j,k}^{(1)} \ \mbox{for all} \ k\in \mathbb{Z}^n \ \mbox{with} \ \langle k, \omega\rangle = 0\, ,$$ see  also \eqref{eq:firstordernfformal}, while the formula for the $C_{j,k}^{(2)}$ is more involved. If it so happens that $\overline{F}^{(1)}=0$, that is, if $F^{(1)}$ does not contain any resonant terms, then \eqref{psievolution} reduces to
  \begin{equation*}
      \dot \psi_j = \omega_j + \varepsilon^2 \overline{F}_j^{(2)}(\psi) + \varepsilon^3\overline{R}_j^{(3)}(\phi, \varepsilon)\, .
  \end{equation*}
  For this specific case, we provide a formula for the $C_{j,k}^{(2)}$  in Theorem \ref{second_order_nf_thm}. 
As was the case for the first-order method, our second-order resonant reconstruction method exploits the approximate linear relation between  observed phase drifts and  Fourier  coefficients of the normal form.
Let us again assume that we are  given reasonable estimates of the coupling strength $\varepsilon$ and
the oscillator frequencies $\omega_1, \ldots \omega_n$,  
and that we have access to noisy phase signals 
$$\Phi^{(m)}(t)\in\mathbb{R}^n \ \mbox{with}\ 1\leq m \leq M\, ,$$
 this time at times $$t=0 \ \mbox{and} \ t=T:=\frac{1}{\varepsilon\sqrt{\varepsilon}}\, .$$
As before, the noise in these signals is assumed small and bounded, that is,  there are constants $L_j>0$ such that for every $1\leq m \leq M$, there is a (lift of an) exact solution $\phi^{(m)}(t)\in \mathbb{R}^n$ to  \eqref{equationsofmotionintro} 
with 
\begin{equation}\label{eq:approxphiagain}
| \phi^{(m)}_j(0)-\Phi^{(m)}_j(0)| \leq L_j\varepsilon\ \mbox{and} \ | \phi^{(m)}_j(T)-\Phi^{(m)}_j(T)| \leq L_j\varepsilon\, .
\end{equation}
Our second order reconstruction method can then be summarized as follows: 
\begin{enumerate}
    \item[{\bf 1)}] Fix an oscillator $1\leq j \leq n$ and calculate the rescaled observed drifts
    \[
    {\Delta}^{(m)} :=  \frac{\Phi^{(m)}_j(T) - \Phi^{(m)}_j(0) -  \omega_j T}{\varepsilon^2 T} \in \mathbb{R} \ \ \mbox{for}\  1\leq m \leq M\, .
    \]
    These drifts together form a (column) vector $\Delta\in \mathbb{R}^M$.  

    \item[{\bf 2)}] Select a  finite library of {\it resonant} Fourier labels 
    $$ \mathcal{K} = \{k_1, \ldots, k_K\}    \subset \{k\in \mathbb{Z}^n\,| \, \langle k,  \omega\rangle = 0\} \, ,$$ 
    where $K:=|\mathcal{K}|$, and create the library matrix
    $$ \Theta := \left( \begin{array}{ccc}
   e^{i\langle k_1, \Phi^{(1)}(0) \rangle} & \hdots & e^{i\langle k_K, \Phi^{(1)}(0)\rangle}  \\
    \vdots & \ddots & \vdots \\
    e^{i\langle k_1, \Phi^{(M)}(0)\rangle} & \hdots & e^{i\langle k_K, \Phi^{(M)}(0)\rangle} 
 \end{array} \right)  : \mathbb{C}^{K} \to \mathbb{C}^M\, .$$
We assume that  $\Theta$ is injective. 
    \item[{\bf 3)}] List the resonant Fourier coefficients $C_{j,k}^{(2)}$ (with $k\in \mathcal{K}$) in the  vector $$C_j^{(2)} : = (C^{(2)}_{j, k_1}, \ldots, C^{(2)}_{j, k_K})^T \in\mathbb{C}^K\, .$$ 
    We estimate this vector by the vector
    $$\hat C_j = (\hat C_{j, k_1}, \ldots, \hat C_{j, k_K})^T \in\mathbb{C}^K$$ 
    defined by  
    $$\hat C_j = {\arg\min}_{B\in \mathbb{C}^K} \sum_{m=1}^M  \left|   \Delta^{(m)} - \sum_{i=1}^{K}    B_{k_i} e^{i\langle k_i, \Phi^{(m)}(0)\rangle}   \right|^2 \in \mathbb{C}^K \, . 
    $$
    This minimizer is given by
    \begin{equation}\label{pseudo-inverseformula2nd}
     \hat C_j = \Theta^+ (\Delta) \in \mathbb{C}^K\, .
      \end{equation}
\end{enumerate}
In the next section, we prove that the estimators $\hat C_{j,k}$ are good approximations of the true normal form coefficients $C_{j,k}^{(2)}$ under some natural conditions. This proof is highly similar to the corresponding proof for the first order reconstruction method.

\section{Validity of second order reconstruction}\label{sec:validity2norder}
The main result of this section is Theorem \ref{thm:2ndordermainthm}, which is the analogue of Theorem \ref{general_thm} for the second order resonant reconstruction method. The statement in Theorem \ref{thm:2ndordermainthm}  may appear complicated, but its proof is highly similar to that of Lemmas \ref{simplelemma} and \ref{lem:approxlemma}  and Theorem \ref{general_thm}. 

Note that it is not assumed in Theorem \ref{thm:2ndordermainthm} that   equations  \eqref{equationsofmotionintro} can be brought into normal form exactly (in which case $\overline{F}^{(1)}=\overline{F}^{\mathcal{N}}=0$). Instead, it suffices that the equations can be transformed into normal form approximately. 

\begin{theorem}\label{thm:2ndordermainthm}
  For  $|\varepsilon| \leq \varepsilon_0$  consider the phase equations 
  \begin{equation}\label{eq:eqmotionagain}
  \dot \phi_j = \omega_j + \varepsilon F_j^{(1)}(\phi) + \varepsilon^2 F_j^{(2)}(\phi)  + \varepsilon^3 R_j^{(3)}(\phi, \varepsilon)\ \mbox{for all}\ 1 \leq j \leq n\, .
  \end{equation}
  We assume that there exists a family of continuously differentiable ``approximate normal form transformations''  $\Gamma^{\varepsilon}: \mathbb{T}^n \to \mathbb{T}^n $  with the property that 
   $\phi(t)$ is a solution to \eqref{eq:eqmotionagain} if and only if  $\psi(t)=\Gamma^{\varepsilon}(\phi(t))$ satisfies 
  \begin{equation}\label{NFequations}
  \dot \psi_j = \omega_j + \varepsilon \overline{F}_j^{(1)}(\phi) + \varepsilon^{2}\overline{F}_j^{(2)}(\phi) + \varepsilon^3\overline{R}^{(3)}(\phi,\varepsilon)\, , 
  \end{equation}
 where $\overline{F}^{(1)}_j \in C^{2}(\mathbb{T}^n)$, $\overline{F}^{(2)}_j \in C^{1}(\mathbb{T}^n)$ and $\overline{R}^{(3)}_j \in C^0(\mathbb{T}^n\times[-\varepsilon_0, \varepsilon_0])$.

 Let $\mathcal{K} = \{k_1, \ldots, k_K\} \subset \{ k\in\mathbb{Z}^n\, |\, \langle k, \omega\rangle =0\}$ be a library of resonant Fourier labels, and decompose the second order normal form term as 
 $$\overline{F}_j^{(2)} = \overline{F}_j^{\mathcal{K}} + \overline{F}_j^{\mathcal{R}} + \overline{F}_j^{\mathcal{N}}\ ,$$ 
in which  
\begin{align} \label{eq:FK2}
\overline{F}_j^{\mathcal{K}}(\phi) & := \sum_{k \in \mathcal{K}} C_{j,k}^{(2)} e^{i\langle k, \phi \rangle }\, , \\ \label{eq:FR2}
\overline{F}_j^{\mathcal{R}}(\phi) & := \sum_{\langle k, \omega\rangle = 0,\, k \notin \mathcal{K}} C_{j,k}^{(2)} e^{i\langle k, \phi \rangle }\, , \\ 
\label{eq:FN2}
\overline{F}_j^{\mathcal{N}}(\phi) & := \sum_{\langle k, \omega\rangle \neq 0} C_{j,k}^{(2)} e^{i\langle k, \phi \rangle }\, .
\end{align}
Assume the following on the transformation and approximate normal form: 
\begin{itemize}
\item[{\it i)}] The transformations $\Gamma^{\varepsilon}$ are close to the identity: there are constants $H_1, \ldots, H_n>0$ such that  
\begin{equation}
\label{gammaestimate}
  |\Gamma^{\varepsilon}(\phi)_j-\phi_j| \leq H_j\varepsilon \,  , 
  \end{equation}
   for all $|\varepsilon| \leq \varepsilon_0$, all $\phi\in\mathbb{T}^n$, and all  $1\leq j \leq n.$ 
\item[{\it ii)}]  The first order normal form nearly vanishes: there are constants $G_1, \ldots, G_n>0$ such that 
\begin{equation}
\label{overlineF1bound}
 | \overline{F}_j^{(1)}(\phi)|  \leq \varepsilon^2 G_j \, , 
  \end{equation}
  for all $|\varepsilon| \leq \varepsilon_0$, all $\phi\in\mathbb{T}^n$, and all  $1\leq j \leq n.$ 
  \item[{\it iii)}] The library $\mathcal{K}$ captures the most prominent resonant terms in $\overline{F}^{(2)}$: there are constants $C_1, \ldots, C_n>0$ such that 
  \begin{equation} \label{FR2estimate}
       \left| \int_0^T \overline{F}_j^{\mathcal{R}}(\phi+\omega t)\, dt \right|  \leq C_j \sqrt{\varepsilon} T \, ,  
  \end{equation}
  for all $|\varepsilon| \leq \varepsilon_0$, all $\phi\in\mathbb{T}^n$, and all  $1\leq j \leq n.$
  \item[{\it iv)}] The nonresonant terms in $\overline{F}^{(2)}$ have a bounded dynamical effect: there are constants $D_1, \ldots, D_n>0$ such that 
  \begin{equation}\label{FN2estimate}
  \left| \int_0^T  \overline{F}_j^{\mathcal{N}}(\phi+\omega t)\, dt \right|  \leq D_j \, , 
  \end{equation}
  for all $|\varepsilon| \leq \varepsilon_0$, all $\phi\in\mathbb{T}^n$, and all  $1\leq j \leq n.$
\end{itemize}
  Denote by  
    $$C_{j}^{(2)} = (C_{j,k_1}^{(2)}, \ldots, C_{j,k_K}^{(2)})^{T} \in \mathbb{C}^K\ \mbox{and}\  \hat C_{j} = (\hat C_{j,k_1}, \ldots, \hat C_{j,k_K})^{T} \in \mathbb{C}^K$$     respectively the vector of true resonant Fourier coefficients $C_{j,k}^{(2)}$  in $\overline{F}_j^{(2)}$, and the vector of their estimates $\hat C_{j,k}$ provided by our second order resonant reconstruction method. Let $\mathfrak{C}_j$ be the constant given in \eqref{eq:Cjdef2norder} below, which is independent of $\varepsilon, K$ and $M$.
Then
\begin{enumerate}
    \item The $l_1$-distance between $\hat C_j$ and $C_j^{(2)}$ satisfies 
    $$\|\hat C_j - C_j^{(2)}\|_1 \leq  M \mathfrak{C}_j \|\Theta^+\|_1 \sqrt{\varepsilon}\, ; $$ 
    \item The $l_2$-distance between $\hat C_j$ and $C_j^{(2)}$ satisfies
    $$\|\hat C_j - C_j^{(2)}\|_2 \leq  \sqrt{M}\mathfrak{C}_j \|\Theta^+\|_2 \sqrt{\varepsilon}\, ;$$ 
    \item The $l_{\infty}$-distance between $\hat C_j$ and $C_j^{(2)}$ satisfies
    $$\|\hat C_j - C_j^{(2)}\|_{\infty} \leq  \mathfrak{C}_j \|\Theta^+\|_{\infty} \sqrt{\varepsilon}\, .$$ 
\end{enumerate}
  \end{theorem}
\begin{proof}
    Recall the assumption in \eqref{eq:approxphiagain} that $\Phi^{(m)}(0)$ and $\Phi^{(m)}(T)$ are observations of true solutions $t\mapsto \phi^{(m)}(t)$ to \eqref{eq:eqmotionagain}  that satisfy
$$| \phi^{(m)}_j(0)-\Phi^{(m)}_j(0)| \leq L_j\varepsilon\ \mbox{and} \ | \phi^{(m)}_j(T)-\Phi^{(m)}_j(T)| \leq L_j\varepsilon\, .$$ 
For each such true solution $\phi(t)$, let us define $\psi(t):=\Gamma^{\varepsilon}(\phi(t))$. Then it follows from \eqref{gammaestimate} that 
\begin{equation}\label{eq:phipsiestimate}
|\psi_j^{(m)}(t)-\phi^{(m)}_j(t)| \leq H_j\varepsilon\ \mbox{for all} \ t\, ,
\end{equation}
and hence also that 
\begin{equation}\label{eq:Phiphiestimate}
|\psi_j^{(m)}(t)-\Phi^{(m)}_j(t)| \leq (L_j+H_j)\varepsilon \ \ \mbox{for} \ t\in \{0,T\} \, .
\end{equation}
We can now define the true rescaled normal form drifts 
\begin{equation}\label{deltapsidef}
\delta^{(m)} := \frac{\psi_j^{(m)}(T) - \psi_j^{(m)}(0) - \omega_j T}{\varepsilon^2 T}\, .
\end{equation}
The observed rescaled drifts are good approximations of these true rescaled drifts, as it follows immediately from \eqref{eq:Phiphiestimate} that  
\begin{equation}
    \label{eq:deltaDeltaestimateagain}
| \Delta^{(m)} - \delta^{(m)}| \leq \frac{2 (L_j+H_j)}{\varepsilon T}\, . 
\end{equation}
Note that this expression is of the order $\sim \sqrt{\varepsilon}$ for $T\sim \frac{1}{\varepsilon\sqrt{\varepsilon}}$. 

We now proceed to express $\delta^{(m)}$ in terms of the Fourier coefficients $C_{j,k}^{(2)}$ of the second order normal form, just like we did for the first order reconstruction method in Sections \ref{sec:approxdriftdyn} and \ref{sec:validitysection}. 
In order to do this, note that integration of \eqref{NFequations} from $t=0$ to $t=T$
yields that 
\begin{align}\notag 
 \psi_j^{(m)}(T) \, -\, & \psi_j^{(m)}(0)  -\omega_j T =   \ \varepsilon \int_0^T \overline{F}_j^{(1)}(\psi^{(m)}(t))dt \  \\  \label{eq:NFintegral}
& + \varepsilon^2 \int_0^T \overline{F}_j^{(2)}(\psi^{(m)}(t))dt + \varepsilon^3\int_0^T \overline{R}_j^{(3)}(\psi^{(m)}(t), \varepsilon)dt\, .  
\end{align}
We rewrite and estimate the three integral terms on the right hand side of \eqref{eq:NFintegral}. The first and third integral terms in \eqref{eq:NFintegral} are easy to bound: 
\begin{align}
\label{eq:1integralbound}
&  \left| \varepsilon \int_0^T \overline{F}_j^{(1)}(\psi^{(m)}(t))dt\right| \leq \varepsilon^3 T G_j \,  , \\
\label{eq:3integralbound}
& \left| \varepsilon^3\int_0^T \overline{R}_j^{(3)}(\psi^{(m)}(t), \varepsilon)dt\right| \leq \varepsilon^3 T\|\overline{R}^{(3)}_j\|_0\, ,
\end{align}
where \eqref{eq:1integralbound} clearly follows from \eqref{overlineF1bound}. 
We investigate the second integral term in \eqref{eq:NFintegral} by mimicking the proof of Lemmas \ref{simplelemma} and \ref{lem:approxlemma}.  The  first step is to note that integration of \eqref{NFequations} from $s=0$ to $s=t$ gives
$$\psi_i^{(m)}(t)=\psi_i^{(m)}(0)+\omega_i t + r_i(t, \varepsilon)\ , $$
for certain functions $r_i(t,\varepsilon)$  satisfying 
\begin{align} \notag
| r_i(t,\varepsilon) | & \leq  \varepsilon^2t \left(  \varepsilon^{-1} \|\overline{F}_i^{(1)}\|_0 + \| \overline{F}_i^{(2)}\|_0 + \varepsilon \|  \overline{R}_i^{(3)}\|_0 \right) \\
& \leq \varepsilon^2t \left( \varepsilon G_i + \| \overline{F}_i^{(2)}\|_0 + \varepsilon \|  \overline{R}_i^{(3)}\|_0 \right) \, .
\end{align}
 As in the proof of Lemma \ref{simplelemma}, we can therefore estimate
 \begin{align}\notag 
 & \left| \varepsilon^2 \int_0^T \overline{F}_j^{(2)}(\psi^{(m)}(t))dt - \varepsilon^2 \int_0^T \overline{F}_j^{(2)}(\psi^{(m)}(0) + \omega t) dt \right| \\ \label{eq:yetanotherestimate}
 & \leq \varepsilon^2\int_0^T \sum_{i=1}^n \|\partial_i\overline{F}_j^{(2)}\|_0 |r_i(t,\varepsilon)| dt \leq A_j \varepsilon^4T^2\, .  
 \end{align}
in which
$$A_j = \frac{1}{2}\sum_{i=1}^n \|\partial_i\overline{F}_j^{(2)}\|_0\left(  \varepsilon_0 G_i + \| \overline{F}_i^{(2)}\|_0 + \varepsilon_0 \|  \overline{R}_i^{(3)}\|_0 \right) \, . $$
Next, in view of \eqref{eq:Phiphiestimate}, 
\begin{align}\notag
   & \left| \varepsilon^2 \int_0^T \overline{F}_j^{(2)}(\psi^{(m)}(0) + \omega t) dt -  \varepsilon^2 \int_0^T \overline{F}_j^{(2)}(\Phi^{(m)}(0) + \omega t) dt\right|  \\ 
   &  \leq B_j \varepsilon^3T  \ \ \mbox{in which} \ B_j:=  \sum_{i=1}^n \|\partial_i \overline{F}_j^{(2)}\|_0 (L_i+H_i) \, .
   \label{yetyetanotherestimate}
\end{align}
We continue by mimicking the proof of Lemma \ref{lem:approxlemma}. Recalling that $\overline{F}_j^{(2)} = \overline{F}_j^{\mathcal{K}} + \overline{F}_j^{\mathcal{R}} +\overline{F}_j^{\mathcal{N}}$, we estimate 
\begin{align}\notag 
& \left| \varepsilon^2 \int_0^T \overline{F}_j^{(2)}(\Phi^{(m)}(0)+\omega t ) dt - \varepsilon^2 T \sum_{k\in \mathcal{K}} C^{(2)}_{j,k}e^{i\langle k,\Phi^{(m)}(0)\rangle} \right| \\ \notag
& \leq \varepsilon^2 \left| \int_0^T \overline F_j^{\mathcal{R}}(\Phi^{(m)}(0)+\omega t) dt \right| + \varepsilon^2\left|\int_0^T \overline F_j^{\mathcal{N}} (\Phi^{(m)}(0)+\omega t) dt\right| \\
& \leq C_j \varepsilon^{\frac{5}{2}}T + \varepsilon^2D_j\, ,
\label{eq:yetyetyetanotherone}
\end{align}
where the inequalities in the last line follow from assumptions \eqref{FR2estimate} and \eqref{FN2estimate}. Combining \eqref{eq:NFintegral} with \eqref{eq:1integralbound}, \eqref{eq:3integralbound}, \eqref{eq:yetanotherestimate}, \eqref{yetyetanotherestimate} and \eqref{eq:yetyetyetanotherone} yields 
\begin{align}
  \notag & \left| \psi^{(m)}_j(T)-\psi^{(m)}_j(0) - \omega_j T -  \varepsilon^2T \sum_{k\in \mathcal{K}} C^{(2)}_{j,k}e^{i\langle k,\Phi^{(m)}(0)\rangle}  \right| \\ 
   & \leq A_j\varepsilon^4 T^2 + 
\varepsilon^3 T ( B_j + G_j + \|\overline{R}_j^{(3)}\|_0 ) + C_j\varepsilon^{\frac{5}{2}}T + \varepsilon^2D_j\, .
\label{eq:whatamess}
\end{align}
Next, combining \eqref{eq:whatamess} with \eqref{deltapsidef} and \eqref{eq:deltaDeltaestimateagain} we  obtain 
\begin{align}\notag 
  &   \left| \Delta^{(m)}-  \sum_{k\in \mathcal{K}} C^{(2)}_{j,k}e^{i\langle k,\Phi^{(m)}(0)\rangle}  \right| 
   \\
  &  \leq 
   \left| 
\Delta^{(m)} -  \delta^{(m)}\right| + \left| \delta^{(m)}-  \sum_{k\in \mathcal{K}} C^{(2)}_{j,k}e^{i\langle k,\Phi^{(m)}(0)\rangle}  \right| \notag 
\\
& \leq \frac{2(L_j+H_j)}{\varepsilon T} + A_j\varepsilon^2T + \varepsilon(B_j + G_j + \|\overline{R}_j^{(3)}\|_0 ) +C_j\sqrt{\varepsilon}+D_j/T\, .
\label{eq:almostthere}
\end{align}
For $T= \frac{1}{\varepsilon\sqrt{\varepsilon}}$, the expression on the right hand side of \eqref{eq:almostthere} simplifies  to 
$$\left( 2(L_j+H_j) +A_j +C_j + \sqrt{\varepsilon}(B_j + G_j + \|\overline{R}_j^{(3)}\|_0 ) + \varepsilon D_j \right) \sqrt{\varepsilon} \leq \mathfrak{C}_j\sqrt{\varepsilon}$$
in which  
\begin{equation}\label{eq:Cjdef2norder}
 \mathfrak{C}_j  :=  2(L_j+H_j) +A_j +C_j + \sqrt{\varepsilon_0}(B_j + G_j + \|\overline{R}_j^{(3)}\|_0 ) + \varepsilon_0 D_j \, .
\end{equation}
To summarize, we proved that 
\begin{equation}\label{identity}
\Delta^{(m)} = \sum_{k\in \mathcal{K}} C^{(2)}_{j,k}e^{i\langle k,\Phi^{(m)}(0)\rangle}  + \rho^{(m)}\, ,
\end{equation}
where  $\rho^{(m)} \in \mathbb{R}$ satisfies the bound $|\rho^{(m)}| \leq \mathfrak{C}_j\sqrt{\varepsilon}$. In particular, the vector $\rho\in\mathbb{R}^M$ satisfies  equations \eqref{rho1bound}, \eqref{rho2bound} and \eqref{rhoinfbound}.

One can write equations \eqref{identity} as an equality of $M$-vectors 
$$\Delta = \Theta (C_j^{(2)}) + \rho\, .$$
Now we recall that our second order resonant reconstruction method chooses the estimated normal form coefficients to be $\hat C_j = \Theta^+(\Delta)$, see \eqref{pseudo-inverseformula2nd}. As a result, 
$$\hat C_j - C_j^{(2)} = \Theta^+(\Delta) -\Theta^+\Theta(C_j^{(2)}) = \Theta^+(\Delta) - \Theta^+(\Delta) + \Theta ^+(\rho) = \Theta ^+(\rho) \, .$$
 The estimates on the norms $\|\hat C_j - C_j^{(2)}\|_{1,2,\infty}$  now follow from the argument that was  already given in the proof of Theorem \ref{general_thm}.
\end{proof}

\section{A second order example}\label{sec:2ndorderexample}
To illustrate the second order resonant reconstruction method introduced in Section \ref{sec:2ndordersection}, we consider the  phase oscillator network
\begin{equation}
\begin{aligned}
\dot\phi_1&= 1 +\varepsilon\,\sin(\phi_5-\phi_1)\, ,\\
\dot\phi_2&=3+\varepsilon\left(\sin(\phi_1-\phi_2)+\sin(\phi_3-\phi_2)\right)\, ,\\
\dot\phi_3&=1+\varepsilon\,\sin(\phi_4-\phi_3)\, ,\\
\dot\phi_4&=5\, ,\\
\dot\phi_5&=3+\varepsilon\left(\sin(\phi_3-\phi_5)+\sin(\phi_4-\phi_5)\right)\, ,
\end{aligned}
\label{eq:second_order_example_system}
\end{equation}
of five weakly coupled phase oscillators with frequency vector
\begin{equation}
\omega=\left(1, 3, 1, 5, 3\right).
\end{equation}
We depict this network in Figure \ref{fig:second_order_network}.  
Note that the first order normal form of  equations \eqref{eq:second_order_example_system} vanishes, as none of the interactions on the right hand side of \eqref{eq:second_order_example_system} is resonant. We therefore do not attempt to reconstruct equations \eqref{eq:second_order_example_system}. Instead, we try to numerically approximate the second order resonant normal form of \eqref{eq:second_order_example_system}  from noisy observations of its solutions.
\begin{figure}[H]
    \centering
    \includegraphics[width=0.5\textwidth]{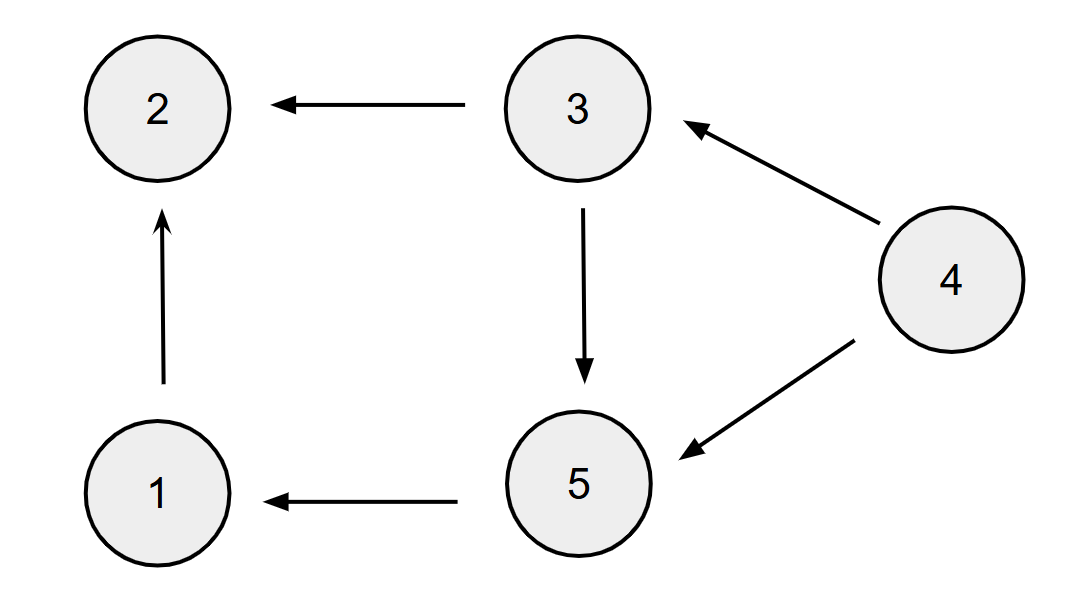}
    \caption{Visualization of the network structure of equations  \eqref{eq:second_order_example_system}.}
    \label{fig:second_order_network}
\end{figure}
\noindent In the selection of the reconstruction library for each oscillator, we anticipate that the second order part $\overline{F}_j^{(2)}$ of the normal form will be a linear combination of complex exponential functions (sines and cosines) that depend on resonant combination angles of the form $$\pm(\phi_j-\phi_i)\pm (\phi_i-\phi_k)\ \mbox{for certain}\ 1\leq i,k \leq n \, .$$
For example, for $j=1$, the only such resonant combination angles are $0, \pm(\phi_1-\phi_3), \pm(2\phi_1-2\phi_3), \pm(\phi_2-\phi_5),  \pm(\phi_1-2\phi_2+\phi_4)$,  and $\pm(\phi_1+\phi_4-2\phi_5)$. Similarly, for $j=2$, the only possibilities are $0, \pm(\phi_1-\phi_3)$, $\pm(\phi_2-\phi_5)$, $\pm(2\phi_2-2\phi_5)$, $\pm(\phi_1-2\phi_2+\phi_4)$, and $\pm(2\phi_2-\phi_3-\phi_4)$. Etc. Based on this principle, we choose the reconstruction libraries 
\begin{align*}
\mathcal{K}_1 = &  \{\alpha_0, \pm \alpha_1, \pm\alpha_2, \pm\alpha_3, \pm \alpha_5, \pm \alpha_6 \}\, ,\\
\mathcal{K}_2 = &  \{\alpha_0, \pm \alpha_1, \pm\alpha_3, \pm\alpha_4, \pm \alpha_5, \pm \alpha_7 \}\, ,\\
\mathcal{K}_3 = &  \{\alpha_0, \pm \alpha_1, \pm\alpha_2, \pm\alpha_3, \pm \alpha_7, \pm\alpha_8 \}\, ,\\
\mathcal{K}_4 = &  \{\alpha_0, \pm \alpha_1, \pm\alpha_3, \pm\alpha_5, \pm \alpha_6, \pm\alpha_7, \pm \alpha_8 \}\, ,\\
\mathcal{K}_5 = &  \{\alpha_0, \pm \alpha_1, \pm\alpha_3, \pm\alpha_4, \pm\alpha_6 \pm \alpha_8 \}\, ,
\end{align*}
in which 
\begin{align} \notag
\alpha_0 &= 0\, , \\ \notag 
\alpha_1 &= e_1 - e_3\, ,\\ \notag
\alpha_2 &= 2e_1 - 2e_3\, ,\\ \notag
\alpha_3 &= e_2 - e_5\, ,\\ \notag
\alpha_4 &= 2e_2 - 2e_5\, ,\\ \notag
\alpha_5 &= e_1- 2e_2 + e_4 \, ,\\ \notag
\alpha_6 &= e_1 + e_4 - 2e_5\, ,\\ \notag
\alpha_7 &= 2e_2 - e_3  - e_4 \, ,\\ \notag
\alpha_8 &= e_3 + e_4 - 2e_5\, . 
\label{eq:alpha_labels_second_order}
\end{align} 
The Fourier labels $\alpha_1, \ldots, \alpha_8$ represent the possible dyadic and triadic interactions in the reconstructed second order normal form. A visualization of the interpretation of these interactions is given in Figure~\ref{fig:second_order_library}. We also note that the libraries together contain a total of $57$ elements. Our reconstruction method will thus determine $57$ different Fourier coefficients.
\begin{figure}[H]
    \centering
    \includegraphics[width=0.8\textwidth]{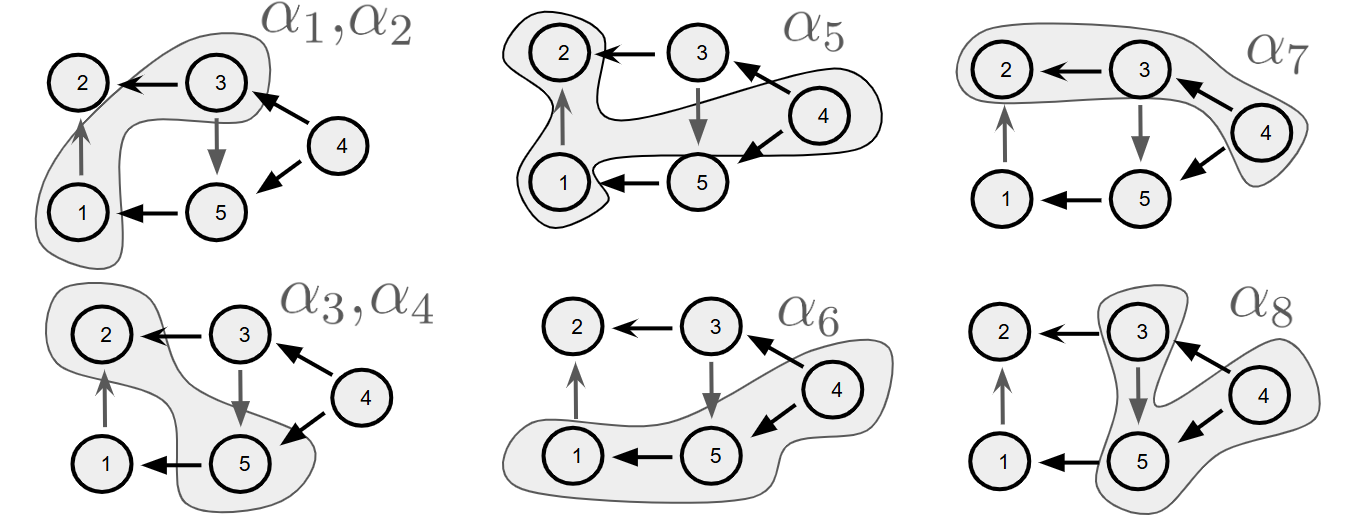}
    \caption{Visualization of the second order resonant Fourier labels $\alpha_1,\dots,\alpha_8$ used in the reconstruction libraries for equations \eqref{eq:second_order_example_system}. }
    \label{fig:second_order_library}
\end{figure}

\noindent We numerically reconstruct the second-order normal form of \eqref{eq:second_order_example_system}  by setting $\varepsilon=\frac{1}{100}$,  and 
choosing $M=50$ initial phases $\phi^{(m)}(0)$ randomly from the  uniform distribution on $\mathbb{T}^5$. Equations \eqref{eq:second_order_example_system} are then integrated in MATLAB starting from these initial conditions, over a time-interval of the length $T=\frac{1}{\varepsilon\sqrt{\varepsilon}}=1000
$, to produce the final phases $\phi^{(m)}(T)$. A small uniformly random perturbation is also added to $\phi^{(m)}(0)$ and $\phi^{(m)}(T)$,  producing the noisy phases $\Phi^{(m)}(0)$ and $\Phi^{(m)}(T)$. The maximal amplitude of this perturbation is chosen to be $\varepsilon$, that is, $L_j=1$ for all $1\leq j\leq n$. Finally, we estimate the Fourier coefficients of the second order normal form by using the recipe in Section \ref{sec:2ndordersection}. We present the resulting reconstructed resonant Fourier coefficients in Table \ref{tab:second_order_complex_coefficients}. Note that the difference between the reconstructed coefficients and the coefficients  of the true normal form of \eqref{eq:second_order_example_system} is much smaller than $\sqrt{\varepsilon}=\frac{1}{10}$.   Our reconstruction method thus successfully  reconstructs the  second order normal form, which we derive analytically in Proposition \ref{prop:2ndordernfexample} below. 
\begin{longtable}{|c|c|c|c|c|}
\hline
\! \textbf{Coefficient} \! & \! \textbf{Real part} \! & \! \textbf{Imag. part} \! & \! \textbf{True value}\! & \!  \textbf{Error} \! \\
\hline
\endfirsthead

\hline
$\hat C_{j,\alpha_k}^{(2)}$ & $\Re(\hat C_{j,\alpha_k}^{(2)})$ & $\Im(\hat C_{j,\alpha_k}^{(2)})$ & $C_{j,\alpha_k}^{(2)}$ & Error \\
\hline
\endhead

$\hat C_{1,\alpha_0}^{(2)}$ & 0.2350 & 0.0000 & 0.25 & 0.0150 \\
$\hat C_{1,\alpha_1}^{(2)}$ & 0.1218 & 0.0049 & 0.125 & 0.0058 \\
$\hat C_{1,\alpha_2}^{(2)}$ & 0.0133 & -0.0104 & 0 & 0.0169 \\
$\hat C_{1,\alpha_3}^{(2)}$ & -0.0111 & -0.0111 & 0 & 0.0157 \\
$\hat C_{1,\alpha_5}^{(2)}$ & -0.0024 & -0.0046 & 0 & 0.0052 \\
$\hat C_{1,\alpha_6}^{(2)}$ & -0.1159 & -0.0072 & -0.125 & 0.0116 \\
\hline
$\hat C_{2,\alpha_0}^{(2)}$ & -0.4727 & 0.0000 & -0.5 & 0.0273 \\
$\hat C_{2,\alpha_1}^{(2)}$ & -0.2492 & 0.0045 & -0.25 & 0.0045 \\
$\hat C_{2,\alpha_3}^{(2)}$ & -0.1165 & 0.0012 & -0.125 & 0.0085 \\
$\hat C_{2,\alpha_4}^{(2)}$ & 0.0058 & 0.0151 & 0 & 0.0161 \\
$\hat C_{2,\alpha_5}^{(2)}$ & 0.0054 & 0.0175 & 0 & 0.0183 \\
$\hat C_{2,\alpha_7}^{(2)}$ & -0.0033 & -0.0105 & 0 & 0.0110 \\
\hline
$\hat C_{3,\alpha_0}^{(2)}$ & 0.1308 & 0.0000 & 0.125 & 0.0058 \\
$\hat C_{3,\alpha_1}^{(2)}$ & -0.0032 & 0.0076 & 0 & 0.0083 \\
$\hat C_{3,\alpha_2}^{(2)}$ & -0.0011 & -0.0097 & 0 & 0.0097 \\
$\hat C_{3,\alpha_3}^{(2)}$ & -0.0068 & 0.0023 & 0 & 0.0072 \\
$\hat C_{3,\alpha_7}^{(2)}$ & -0.0053 & 0.0023 & 0 & 0.0058 \\
$\hat C_{3,\alpha_8}^{(2)}$ & -0.0083 & 0.0023 & 0 & 0.0086 \\
\hline
$\hat C_{4,\alpha_0}^{(2)}$ & 0.0015 & 0.0000 & 0 & 0.0015 \\
$\hat C_{4,\alpha_1}^{(2)}$ & -0.0077 & 0.0023 & 0 & 0.0080 \\
$\hat C_{4,\alpha_3}^{(2)}$ & 0.0095 & 0.0048 & 0 & 0.0106 \\
$\hat C_{4,\alpha_5}^{(2)}$ & -0.0010 & 0.0036 & 0 & 0.0037 \\
$\hat C_{4,\alpha_6}^{(2)}$ & -0.0038 & -0.0039 & 0 & 0.0055 \\
$\hat C_{4,\alpha_7}^{(2)}$ & -0.0005 & 0.0068 & 0 & 0.0068 \\
$\hat C_{4,\alpha_8}^{(2)}$ & 0.0035 & -0.0054 & 0 & 0.0064 \\
\hline
$\hat C_{5,\alpha_0}^{(2)}$ & -0.0205 & 0.0000 & 0 & 0.0205 \\
$\hat C_{5,\alpha_1}^{(2)}$ & 0.0239 & 0.0031 & 0 & 0.0241 \\
$\hat C_{5,\alpha_3}^{(2)}$ & -0.0045 & 0.0058 & 0 & 0.0073 \\
$\hat C_{5,\alpha_4}^{(2)}$ & 0.0007 & 0.0151 & 0 & 0.0151 \\
$\hat C_{5,\alpha_6}^{(2)}$ & -0.0231 & 0.0167 & 0 & 0.0285 \\
$\hat C_{5,\alpha_8}^{(2)}$ & 0.0101 & -0.0112 & 0 & 0.0151 \\
\hline
\caption{Reconstructed Fourier coefficients of the second order normal form of \eqref{eq:second_order_example_system} with $\varepsilon=\frac{1}{100}$, based on $M=50$ noisy observations. The true normal form coefficients are included for comparison.}
\label{tab:second_order_complex_coefficients}
\end{longtable}

\begin{proposition}\label{prop:2ndordernfexample}
    The second order normal form of \eqref{eq:second_order_example_system} is given by 
    \begin{align} \notag
\dot\phi_1&=1+\varepsilon^2 \left(\frac14 +\frac14\cos(\phi_1-\phi_3) -\frac14\cos(\phi_1 + \phi_4-2\phi_5)
 \right)+ \mathcal{O}(\varepsilon^3)\,  ,\\
\notag
\dot\phi_2&=3+\varepsilon^2 \left( -\frac12 -\frac12\cos(\phi_1-\phi_3) -\frac14\cos(\phi_2-\phi_5)
\right) + \mathcal{O}(\varepsilon^3) \, ,\\
\notag
\dot\phi_3&=1+\frac{1}{8} \varepsilon^2  + \mathcal{O}(\varepsilon^3) \, ,\\
\notag
\dot\phi_4&=5 + \mathcal{O}(\varepsilon^3) \, , \\
\notag
\dot\phi_5&=3+ \mathcal{O}(\varepsilon^3)\, .
\end{align}
In particular, the only nonzero complex Fourier coefficients $C^{(2)}_{j,k}$ in the second order normal form are given by 
\begin{align*}
& C_{1,\alpha_0}^{(2)} = \frac14, \ C_{1,\pm \alpha_1}^{(2)} = \frac18,\ C_{1,\pm \alpha_6}^{(2)} = -\frac18\,  ,
\\
& C_{2,\alpha_0}^{(2)} = -\frac12,\ C_{2,\pm \alpha_1}^{(2)} = -\frac14,\ C_{2,\pm \alpha_3}^{(2)} = -\frac18\, ,
\\
& C_{3,\alpha_0}^{(2)}=\frac18\, .
\end{align*}
\end{proposition}
\begin{proof}
    The procedure to compute the second order normal form is explained in Appendix \ref{app:2ndorderNF}. To summarize, the first order normal form transformation (which removes all first order terms from the differential equations) is generated by the vector field $H^{(1)}$ given by
    \begin{align}\notag 
H^{(1)}(\phi)
=
\left( \begin{array}{c} 
\frac12 \cos(\phi_1-\phi_5) \\ 
-\frac12\cos(\phi_2-\phi_1)-\frac12\cos(\phi_3-\phi_2) \\
\frac14\cos(\phi_4-\phi_3) \\
0\\
-\frac12\cos(\phi_5-\phi_3)+\frac12\cos(\phi_5-\phi_4)
\end{array}
\right)\, .
\end{align}
Indeed, one can check that ${\rm ad}_{\omega}(H^{(1)}) = -D_{\omega}H^{(1)} =F^{(1)}$, where $F^{(1)}$ is the order $\sim \varepsilon$ part of the right hand side of \eqref{eq:second_order_example_system}. 
The  order $\sim \varepsilon^2$ part of the normal form is now given by the resonant terms in the Lie bracket $\frac{1}{2}[H^{(1)}, F^{(1)}]$, see Appendix \ref{app:2ndorderNF}. This Lie bracket is computed to be
$$\left( \!\!\! \begin{array}{c}
-\frac14\cos(\phi_1 + \phi_4-2\phi_5)
+\frac14\cos(\phi_1-\phi_3)
+\frac14 
\\
-\frac12\cos(\phi_1-\phi_3)
-\frac{3}{16}\cos(\phi_2- \phi_4)
-\frac14\cos(\phi_2- \phi_5)
+\frac{1}{16}\cos(\phi_2-2\phi_3+\phi_4)
-\frac12
\\
\frac18 \\
0
\\
-\frac{3}{16}\cos(\phi_5-\phi_4)
+\frac{1}{16}\cos(2\phi_3 -\phi_4-\phi_5)
\end{array}\!\!\! \right)\, .
$$ 
The resonant terms in this expression constitute  the order $\sim \varepsilon^2$ part in the normal form as given in the statement of the proposition. The final statement follows because $\cos \langle k, \phi\rangle  = \frac12 e^{i\langle k, \phi\rangle } + \frac{1}{2}e^{-i\langle k, \phi\rangle }$. 
\end{proof}
\noindent Note that only $11$ of the Fourier coefficients of the second order normal form are  nonzero.

\section{Conclusion}
We presented a first order and a second order reconstruction method for coupled phase equations, that work by fitting the time $T$-map of a resonant normal form to  observations of solutions to the equations of motion. The choice for a library  that consists solely of resonant terms circumvents the problem that phase equations can be hard to distinguish in the presence of noise, and may not even be uniquely defined. A resonant reconstruction is moreover sparse by design. This reduced the risk of overfitting and makes our methods robust to noise and uncertainty in the observations. Under natural conditions, we were able to prove that our reconstruction methods estimate the coefficients of the first and second order normal forms correctly. We also demonstrated the accuracy of our methods by means of several numerical examples.  
\appendix 
\section{Least squares and the pseudo-inverse} \label{app:pseudoinverse}
Let $\Theta: \mathbb{C}^K \to \mathbb{C}^M$ be injective and let $\Delta \in \mathbb{C}^M$. The least squares solution to the equation $\Theta (A)  = \Delta$ is defined as the (unique) minimizer $\hat A\in \mathbb{C}^K$ of the sum of squares
 \[ E(B) :=  \sum_{m=1}^M  | \Delta^{(m)}- (\Theta B)_m |^2\, .
 \]
Here we prove the well-known fact that $\hat A = \Theta^+(\Delta)$. 
\begin{lemma}\label{LSM_thm}
   Assume that $\Theta:\mathbb{C}^{K}\to\mathbb{C}^M$ is injective. Then the minimizer $\hat A = {\rm arg\ \! min}_{B} \, E(B)$
is given by the formula $$\hat A = \Theta^+(\Delta) \in \mathbb{C}^K\, , $$ 
in which $\Theta^+ = \left({\Theta}^H\Theta\right)^{-1}\Theta^{H}$ is the Moore-Penrose pseudo-inverse of $\Theta$, and $\Theta^{H}= \overline{ \Theta}^T$ denotes its Hermitian transpose. 
\end{lemma}
\begin{proof}
   Observe that we can write
\begin{align}\nonumber 
 E(B) & =   (\Delta- \Theta B)^T \overline{(\Delta - \Theta B)}   =  \Delta^T\overline{\Delta} - B^T\Theta^T \overline{\Delta}   -   \Delta^T\overline{\Theta} \, \overline{B} +   B^T \Theta^T \overline{\Theta}\, \overline B.
\end{align}
Differentiating this at $B\in \mathbb{C}^K$ in the direction of $Z\in \mathbb{C}^K$ we find 
 \begin{align}\nonumber 
  \left.\frac{d}{d\varepsilon}\right|_{\varepsilon=0} \!\!\!\!\! E(B+ \varepsilon Z)   \nonumber 
 &= - Z^T \Theta^T\overline{\Delta} -   \Delta^T \overline{\Theta}\, \overline{Z}  +  Z^T\Theta^T\overline{\Theta}\, \overline{B} +   B^T \Theta^T\overline{\Theta}\, \overline{Z}  \\ 
 &=   Z^T \left( (\Theta^T\overline{\Theta}) \overline B - \Theta^T\overline{\Delta} \right) +   \left( B^T(\Theta^T\overline{\Theta}) - \Delta^T \overline{\Theta} \right)\overline{Z}   \nonumber \\ 
 &=   Z^T \overline{ \left( ({\Theta}^H\Theta)B - \Theta^{H}\Delta \right) } +   \overline{Z}^T \left( ({\Theta}^H\Theta)B - \Theta^{H}\Delta \right)\, .   \nonumber 
 \end{align}
This expression must vanish for all $Z$ at the minimizer $B=\hat A$. Choosing $Z$ real, we see that the real part of $({\Theta}^H\Theta)\hat A - \Theta^{H}\Delta$ must vanish, while choosing $Z$ purely imaginary, we find that the imaginary part of the same expression must vanish. It follows that $({\Theta}^H\Theta)\hat A - \Theta^{H}\Delta = 0$. Since $\Theta$ is injective, $\Theta^{H}\Theta$ is invertible,  so 
$\hat A = \Theta^+\Delta  = \left( {\Theta}^{H}\Theta\right)^{-1}\Theta^{H}\Delta$. 
\end{proof} 

\section{Lie transformations}\label{app:Lie}
The normal form of a differential equation is a standardized asymptotic representation of it. This representation is typically achieved by applying coordinate transformations in phase space that bring the equation into a certain prescribed form. These coordinate transformations are usually themselves taken to be the flow of a vector field. Such transformations go under the name ``Lie transformations''. We study Lie transformations in some detail in this appendix.  We shall apply them to compute normal forms of coupled oscillator networks in Appendices \ref{sec:normalformappendix} and \ref{app:2ndorderNF}.

For a (globally) Lipschitz continuous vector field $H: \mathbb{T}^n\to\mathbb{R}^n$ we denote by $\Gamma_H: \mathbb{T}^n \times \mathbb{R} \to \mathbb{T}^n$  its flow, which is defined for all time.  It is well-known that $\Gamma_H$ is $C^k$ when $H$ is $C^k$. We denote by $\Gamma_H^{s}(\phi):=\Gamma_H(\phi, s)$ the time-$s$ flow. It holds that $\Gamma_H^0 ={\rm Id}_{\mathbb{T}^n}$ and $\Gamma_H^{s_1}\circ \Gamma_H^{s_2}=\Gamma_H^{s_1+s_2}$. Thus, each $\Gamma_H^s$ is a diffeomorphism, with inverse $\Gamma_H^{-s}$. By definition, the family $\Gamma_H^s$ satisfies the differential equation 
\begin{equation} \label{flow_derivativeextra}
    \frac{\partial }{\partial s} \Gamma_H^s(\phi) = H(\Gamma_H^s( \phi))\, .
\end{equation}
 This identity can be used to  compute a Taylor expansion of $\Gamma_H^s$ with respect to $s$. For $C^1$-smooth vector fields $F,G:\mathbb{T}^n\to\mathbb{R}^n$, we shall denote  by $D_F(G): \mathbb{T}^n\to\mathbb{R}^n$  the  derivative of $G$ in the direction of  $F$, i.e., 
 $$D_F(G)(\phi):=DG(\phi)\cdot F(\phi)\, .$$ 
 The statement on the Taylor expansion of $\Gamma_H^s$ is the following.
\begin{proposition}\label{prop:nearidexpansion} Let  $H:\mathbb{T}^n\to\mathbb{R}^n$ be Lipschitz continuous. Then 
\begin{equation}\label{eq:Gammaapprox1storder}
\Gamma_H^s(\phi) = \phi + r^{(1)}(\phi, s) \ \mbox{for}\ r^{(1)}(\phi,s ) := \int_0^sH(\Gamma_H^{\sigma}(\phi))\, d\sigma\, .
\end{equation}
When $H$ is $C^k$ for some $k\geq 1$, then $s\mapsto \Gamma_H^s(\phi)$ is $C^{k+1}$ and 
\begin{equation}\label{eq:Gammaapprox1kthorder}
    \Gamma_H^s(\phi)= \phi + \sum_{m=1}^{k} \frac{s^m}{m!} D_H^{m-1}(H) (\phi) + r^{(k+1)}(\phi, s) \, ,
    \end{equation}
    in which $r^{(k+1)}(\phi, s) = \frac{1}{k!}\int_0^s D_H^k(H)(\Gamma_H^{\sigma} (\phi))(s-\sigma)^kd\sigma$.
\end{proposition}
\begin{proof}
    The Picard-Lindel\"of theorem guarantees that $s\mapsto \Gamma_H^s(\phi)$ is continuous. It thus follows by induction from \eqref{flow_derivativeextra} that $s\mapsto \Gamma_H^s(\phi)$ is $C^{k+1}$ when $H$ is $C^k$. In particular, equation \eqref{flow_derivativeextra} is well-defined. Integrating it from $0$ to $s$ gives \eqref{eq:Gammaapprox1storder}. 
     Next, we claim that
\begin{equation}\label{eq:derivativeGamma}
        \frac{\partial^m}{\partial s^m}\Gamma_H^s(\phi) = D_H^{m-1}(H)(\Gamma_H^s(\phi))\ \mbox{for}\  1\leq m\leq k+1\, .
        \end{equation}
        Indeed, for $m=1$ equation \eqref{eq:derivativeGamma} reduces to \eqref{flow_derivativeextra}. Moreover, assuming \eqref{eq:derivativeGamma} for $1\leq m \leq k$, differentiation of \eqref{eq:derivativeGamma} with respect to $s$ yields  $$\frac{\partial^{m+1}}{\partial s^{m+1}}\Gamma_H^s(\phi) = D(D^{m-1}_H(H))(\Gamma_H^s(\phi))\cdot H(\Gamma_H^s(\phi))  = D_H^{m}(H) ( \Gamma_H^s(\phi))\, .$$ 
    Thus, \eqref{eq:derivativeGamma} follows by induction. In particular, we find that $$\left. \frac{\partial^m}{\partial s^m}\Gamma_H^s(\phi)\right|_{s=0} = D_H^{m-1}(H)(\phi)\, .$$ The statement of the Proposition now follows from Taylor's theorem with integral remainder.  
\end{proof}
\begin{remark}\label{rem:flowremark}
    The integral equation \eqref{eq:Gammaapprox1storder} can be used to inductively compute bounds on the uniform norms of $\Gamma_H^s$ and its spatial derivatives. Most importantly, it is clear from \eqref{eq:Gammaapprox1storder} that $$\sup_{\phi\in\mathbb{T}^n} \|\Gamma_H^s(\phi)_j-\phi_j\| \leq s \|H_j\|_0\, .$$ 
    Bounds on the spatial derivatives of the flow $\Gamma_H^s$ can be obtained in the usual way: when $H$ is $C^1$, then so is $\Gamma_H^s$ and one can  differentiate \eqref{eq:Gammaapprox1storder} to $\phi$. The resulting equations can be used to derive integral inequalities for the uniform norms of the spatial derivatives of $\Gamma_H^s$,   which can then be solved using Gronwall's lemma. One can then continue estimating the second spatial derivatives of $\Gamma_H^s$, etc. We omit the details.
\end{remark}
\begin{remark}\label{remark:transformedVariableBound}
We may estimate the remainder in Proposition \ref{prop:nearidexpansion} by 
$$| r^{(k+1)}(\phi, s) | \leq    \frac{s^{k+1} \|D_H^k(H)\|_{C^0}}{(k+1)!}\ \mbox{for}\ k=0,1,\ldots \, . $$
Using that $D_H^k(H)=D(\cdots D(DH\cdot H)\cdot H)\cdot H$, the norm $\|D_H^k(H)\|_{0}$ can in turn be bounded in terms of the $C^k$-norm  of $H$. 
\end{remark}
\noindent 
We now investigate how the transformations $\Gamma_H^s$ transform a vector field $F$ on $\mathbb{T}^n$. Let $s\mapsto \phi(s)$ be an integral curve of $F$, so $\frac{d\phi(s)}{ds} = F(\phi(s))$, and let   $\Gamma: \mathbb{T}^n\to\mathbb{T}^n$ be a diffeomorphism.  Then the curve $ \psi(s) :=\Gamma(\phi(s))$ satisfies 
$$\frac{d \psi(s)}{ds} = D\Gamma(\phi(s))\cdot F(\phi(s)) = D\Gamma(\Gamma^{-1}(\psi(s)))\cdot F(\Gamma^{-1}(\psi(s))) \, .$$
This means that $\psi$ is in integral curve of  the {\it pushforward} vector field $\Gamma_*F: \mathbb{T}^n\to\mathbb{R}^n$  defined by 
$$\Gamma_*F = (D\Gamma \cdot F)\circ \Gamma^{-1} \, .$$
We want to compute the Taylor expansion with respect to $s$ of $(\Gamma_H^s)_*F$. We first recall the following definition
\begin{definition}
Let $F,G:\mathbb{T}^n\to\mathbb{R}^n$ be $C^2$-smooth vector fields. Then their Lie bracket $[F,G]$ is defined to be the vector field 
\begin{align}\label{lieDef}
[F,G] := \at{\frac{d}{ds}}{s=0}(\Gamma_F^s)_*G
\end{align}
It is given by the formula $[F,G]=D_G(F)-D_F(G)=DF\cdot G -DG\cdot F$. 
\end{definition}
\noindent In the following proposition, we denote by ${\rm ad}_H$ the linear operator of taking the Lie bracket with $H$, that is, ${\rm ad}_H(F):=[H,F]$.
\begin{proposition}\label{transformationvectorfield}
Let $F$ be $C^k$ and $H$ be $C^{k+1}$ for some $k\geq 1$. Then  $s \mapsto (\Gamma_H^s)_*F(\phi)$ is $C^{k}$ and 
$$(\Gamma_H^s)_*F(\phi) = \sum_{m=0}^{k-1}\frac{s^m}{m!}{\rm ad}_{H}^m(F)(\phi) + \rho^{(k)}(\phi, s)\, ,$$
in which $\rho^{(k)}(\phi,s)= \frac{1}{(k-1)!}\int_0^s{\rm ad}_H^{k}((\Gamma_H^{\sigma})_*F)(\phi)(s-\sigma)^{k-1}d\sigma$.
\end{proposition}
\begin{proof}
Recall that the flow of a $C^{k+1}$ vector field is $C^{k+1}$, meaning that $(\phi,s)\mapsto \Gamma^{s}_H(\phi)$ is $C^{k+1}$. Therefore, the defining formula for $(\Gamma_H^s)_*F$ shows that $s\mapsto (\Gamma_H^s)_*F(\phi)$ is $C^{k}$.  Next, observe that the family of pushforward vector fields $(\Gamma_H^s )_{*}F$ satisfies the linear differential equation
     \begin{align*}
        {\frac{\partial}{\partial t}}(\Gamma_H^t )_{*}F   =\at{\frac{d}{dh}}{h=0} \!\!\! (\Gamma_H^h)_{*}(\Gamma_H^s )_{*}F  =[H,(\Gamma_H^s )_{*}F] = \mathrm{ad}_H((\Gamma_H^s )_{*}F)\, .
    \end{align*}
    The right hand side of this equation is clearly $C^{k-1}$ as a function of $(\phi, s)$. Differentiating this identity to $s$ (a maximum of $k-1$ times) gives us the higher derivatives of $s\mapsto (\Gamma_H^s )_{*}F$:
\begin{equation} \label{inductive_derivative_push_forward}
  \frac{\partial^m}{\partial s^m} (\Gamma_H^s )_{*} F = \mathrm{ad}_H^m ((\Gamma_H^s )_{*} F)\ \mbox{for}\ 1\leq m\leq k\, .
\end{equation}
In particular, 
$$\left.\frac{\partial^m}{\partial s^m}\right|_{s=0} (\Gamma_H^s )_{*} F = {\rm ad}_H^m(F)\  \mbox{for} \ 1\leq m\leq k-1\, .$$ 
The statement of the proposition thus follows from Taylor's theorem with integral remainder.
\end{proof}
\begin{remark}\label{remark:boundRemainderVectorField}
    The remainder term in Proposition \ref{transformationvectorfield} is bounded by
    $$\|\rho^{(k)}(\phi,s)\| \leq \frac{s^{k}}{k!}  \sup_{0\leq \sigma\leq s} \| {\rm ad}_H^k((\Gamma^{\sigma}_H)_*F) \|_{0}\,   .$$
This can in turn be bounded in terms of the $C^{k}$-norm of $F$ and the $C^{k+1}$-norm of  $H$.
\end{remark}

\section{Normal forms}\label{sec:normalformappendix}
We now apply the Lie transformations introduced in Appendix \ref{app:Lie} to compute normal forms of weakly coupled phase oscillator systems of the form 
\begin{equation}\label{main}
\dot{\phi} = F(\phi, \varepsilon) = \omega + \varepsilon F^{(1)}(\phi) + \ldots + \varepsilon^kF^{(k)}(\phi) +  \varepsilon^{k+1}R^{(k+1)}(\phi, \varepsilon) \, .
\end{equation}
Here, $\phi=(\phi_1, \ldots, \phi_n)\in \mathbb{T}^n$,  $\omega\in \mathbb{R}^n$ is a vector of frequencies, and  $\varepsilon$ is a small parameter. Equations \eqref{equationsofmotionintro} in the introduction are a special case. We assume that $F^{(1)}, \ldots, F^{(k)}$ are smooth vector fields on $\mathbb{T}^n$, and also that $R^{(k+1)}:\mathbb{T}^n\times [-\varepsilon_0, \varepsilon_0]\to \mathbb{R}^n$ is smooth and bounded. 

Now let $H^{(1)}:\mathbb{T}^n\to\mathbb{R}^n$ be another smooth vector field on $\mathbb{T}^n$. We consider the Lie transformation $\Gamma^{\varepsilon}_{H^{(1)}}: \mathbb{T}^n\to\mathbb{T}^n$, i.e., the time-$\varepsilon$ flow of $H^{(1)}$.  Setting $s=\varepsilon$ and $H=H^{(1)}$ in Proposition \ref{prop:nearidexpansion}, we find that $$\Gamma_{H^{(1)}}^{\varepsilon}(\phi) = \phi + \varepsilon H^{(1)}(\phi) + \frac{1}{2}\varepsilon^2DH^{(1)}(\phi)H^{(1)}(\phi) + \mathcal{O}(\varepsilon^3)\, .$$ 
More importantly, choosing $s=\varepsilon$ in Proposition \ref{transformationvectorfield} we find that $\phi(t)$ satisfies \eqref{main} if and only if $t\mapsto \Gamma_{H^{(1)}}^{\varepsilon}(\phi(t))$ is an integral curve of the transformed vector field
\begin{align}\notag 
 \left(\Gamma_{H^{(1)}}^{\varepsilon}\right)_*F  & = F + \varepsilon [H^{(1)},F](\phi) + \frac{1}{2}\varepsilon^2[H^{(1)},[H^{(1)},F]](\phi) + \mathcal{O}(\varepsilon^3)  = \\ 
&   \omega + \varepsilon \left( F^{(1)} + [H^{(1)}, \omega] \right)  + \notag \\
& \varepsilon^2\left(F^{(2)} + [H^{(1)}, F^{(1)}] + \frac{1}{2}[H^{(1)},  [ H^{(1)}, \omega]] \right) + \mathcal{O}(\varepsilon^3) \, .\label{eq:transformedFfirst}
\end{align}
We see in particular that the coordinate transformation $\Gamma_{H^{(1)}}^{\varepsilon}$ changes the first order term $F^{(1)}$ of our vector field into $$\overline{F}^{(1)}:= F^{(1)}+[H^{(1)}, \omega] = F^{(1)}- {\rm ad}_{\omega}(H^{(1)})\, .$$ 
By choosing $H^{(1)}$ we may thus subtract any element in the image of ${\rm ad}_{\omega}$ from $F^{(1)}$.
Note that the operator ${\rm ad}_{\omega}$ is given by the simple formula 
$${\rm ad}_{\omega}(H) = [\omega, H] = (D_H\omega) - (D_\omega H) = - (D_\omega H) = -\sum_{i=1}^n \omega_i \frac{\partial{H}}{{\partial \phi_i}}\, ,$$
as any directional derivative of any constant vector field  vanishes.

By applying a second Lie transformation 
$\Gamma_{H^{(2)}}^{\varepsilon^2} = {\rm Id}_{\mathbb{T}^n} + \varepsilon^2H^{(2)} + \mathcal{O}(\varepsilon^4)$, we can further transform our vector field into
\begin{align}\notag 
&  
\left(\Gamma_{H^{(2)}}^{\varepsilon^2}\right)_*\left(\Gamma_{H^{(1)}}^{\varepsilon}\right)_*F  = \notag
  \omega + \varepsilon \overline{F}^{(1)}  +  \varepsilon^2 \overline{F}^{(2)} + \mathcal{O}(\varepsilon^3) \, .
\end{align}
in which 
\begin{equation}\label{eq:secondterm}
\overline{F}_2:=   F^{(2)} +  [H^{(1)}, F^{(1)}]  + \frac{1}{2}[H^{(1)},  [ H^{(1)}, \omega]]  - {\rm ad}_{\omega}(H^{(2)})\, . 
\end{equation}
This follows from Proposition \ref{transformationvectorfield},  substituting $s = \varepsilon^2$ and  $H=H^{(2)}$, and  replacing $F$ by  $\left(\Gamma_{H^{(1)}}^{\varepsilon}\right)_*(F)$ as given in \eqref{eq:transformedFfirst}. Note that applying $\Gamma_{H^{(2)}}^{\varepsilon^2}$ does not change $\overline{F}^{(1)}$, the order-$\varepsilon$ part of $\left(\Gamma_{H^{(1)}}^{\varepsilon}\right)_*(F)$. Moreover, we again see that we may subtract any element in the image of ${\rm ad}_{\omega}$ from  the order-$\varepsilon^2$ part of the vector field by choosing $F=H^{(2)}$ appropriately. 

One can obviously continue transforming the vector field by applying further Lie transformations $\Gamma_{H^{(3)}}^{\varepsilon^3}, \ldots, \Gamma_{H^{(k)}}^{\varepsilon^k}$. The transformation  $\Gamma_{H^{(r)}}^{\varepsilon^r}={\rm Id}_{\mathbb{T}^n} + \varepsilon^rH^{(r)}(\phi) + \mathcal{O}(\varepsilon^{2r})$ clearly will not change the lower order terms $\overline{F}^{(1)}, \ldots, \overline{F}^{(r-1)}$ of the vector field. Moreover, $\Gamma_{H^{(r)}}^{\varepsilon^r}$ will change the order-$\varepsilon^r$ part of the vector field by subtracting a term ${\rm ad}_{\omega}(H^{(r)})$ from it. We do not study the precise effects of these further Lie transformations, as we shall only use transformations of the form $\Gamma_{H^{(1)}}^{\varepsilon}$ and $\Gamma_{H^{(2)}}^{\varepsilon^2}$ in this paper.

Next, we turn to studying the transformed vector field 
\eqref{eq:transformedFfirst} order by order. Recall that the  order-$\varepsilon$ part of this transformed vector field is $\overline{F}^{(1)}=F^{(1)}-{\rm ad}_{\omega}(H^{(1)})$. Our goal is to ``simplify'' this expression  by choosing $H^{(1)}$. For instance, if we can find an $H^{(1)}$ with 
\begin{align}\label{eq:firsthomologicalequation}
{\rm ad}_{\omega}(H^{(1)})=F^{(1)}\, ,
\end{align} then we would have that $\overline{F}^{(1)}=0$. Equation \eqref{eq:firsthomologicalequation} for the unknown $H^{(1)}$ is an example of a so-called {\it homological equation}. To see if it can be solved, let us assume that $F^{(1)}$ and $H^{(1)}$ admit convergent Fourier expansions,  
\begin{align} \label{eq:F1expansion}
    & F^{(1)}_j(\phi) = \sum_{k \in \mathbb{Z}^n} A_{j, k}^{(1)} e^{i \langle k, \phi \rangle} \ , \ \mbox{with} \ A_{j,k}\in \mathbb{C}\, ,  \\  \label{eq:H1expansion}
& H^{(1)}(\phi)  = \sum_{k\in\mathbb{Z}^n} H_{j,k}^{(1)} e^{i\langle k, \phi\rangle}\ , \ \mbox{with} \ H_{j,k}\in \mathbb{C} \, ,
\end{align}
for $1\leq j \leq n$.  
Assuming that $H^{(1)}_j$ can be differentiated term-by-term,  
\begin{align}\label{Fbarexpansion}
\overline{F}^{(1)}_j(\phi) = F^{(1)}_j(\phi) - {\rm ad}_{\omega}(H^{(1)}_j) =  \sum_{k\in \mathbb{Z}^n} \left(A_{j,k}^{(1)} + i\langle k, \omega\rangle H_{j,k}^{(1)}\right) e^{i\langle k, \phi\rangle}\, .
\end{align}
This clearly shows that we cannot solve the homological equation \eqref{eq:firsthomologicalequation}.
 Instead, we can (formally) remove the $k$-th Fourier term in \eqref{Fbarexpansion} only when $\langle k, \omega\rangle \neq 0$, namely by choosing $H_{j,k}^{(1)}=-\frac{A_{j,k}^{(1)}}{i\langle k, \omega\rangle}$. The corresponding vector field $H^{(1)}$ is given by the formal Fourier series 
\begin{equation}\label{eq:H1expansionexplicit}
H^{(1)}_j(\phi) = -\sum_{\langle k, \omega\rangle \neq 0} \frac{A_{j,k}}{i\langle k, \omega\rangle} e^{i\langle k, \phi\rangle}\, .
\end{equation} 
With this choice for $H^{(1)}_j$, we in fact formally obtain
\begin{equation}
    \label{eq:firstordernfformal}
\overline{F}^{(1)}_j(\phi) = \sum_{\langle k, \omega\rangle = 0} A_{j,k} e^{i\langle k, \phi\rangle} \, .
\end{equation}
We note that this $\overline{F}^{(1)}_j$ is a formal sum 
of resonant terms only (i.e., terms of the form $A_{j, k}^{(1)} e^{i\langle k, \phi\rangle}$ with $\langle \omega, k\rangle = 0$). Nonresonant terms (those with $\langle \omega, k\rangle \neq 0$) are removed from $F^{(1)}_j$ by the coordinate transformation $\Gamma_{H^{(1)}}^{\varepsilon}$. Continuing in this way, we can similarly select $H^{(2)}$ so that $\overline{F}^{(2)}$ is a  formal linear combination of resonant terms, etc. We reserve a special name for the vector fields that we obtain in this way:
\begin{definition}
The vector field  
\begin{equation}\label{normalFormDef}
     \overline{F} =\omega + \varepsilon \overline{F}^{(1)} + \varepsilon^2 \overline{F}^{(2)}+ \dots \varepsilon^k \overline{F}^{(k)}+ \varepsilon^{k+1}\overline{R}^{(k+1)}
\end{equation}
on $\mathbb{T}^n$ is said to be in \emph{normal form} to order $k$ if for all $1\leq r\leq k$, it holds that $\overline{F}^{(r)}$ is a  sum of resonant terms only. 
\end{definition}
\noindent The analysis above shows that, formally (!), any vector field on $\mathbb{T}^n$ of the form \eqref{main} can be brought into normal form to any finite order by a finite sequence of Lie transformations, each being the flow of a formal vector field on $\mathbb{T}^n$. For example, we already showed that the first order part $F^{(1)}$ can be transformed into \eqref{eq:firstordernfformal}. The following result shows when this can  be done analytically:
\begin{theorem}\label{first_order_nf_thm}
Assume that 
\begin{equation}\label{assumptionsconvergence}
\sum_{k\in\mathbb{Z}^n} |A_{j,k}^{(1)}| < \infty \ \mbox{and}\ \sum_{\langle k, \omega\rangle \neq 0} \left| \frac{A_{j,k}}{\langle k, \omega\rangle}\right| <\infty\ \mbox{for all} \ 1\leq j \leq n\, . 
\end{equation}
Then the vector fields $F^{(1)}$, $H^{(1)}$,  and $\overline{F}^{(1)}$ on $\mathbb{T}^n$ defined in \eqref{eq:F1expansion}, \eqref{eq:firstordernfformal} and \eqref{eq:H1expansionexplicit} are continuous and satisfy the homological equation 
$$\overline{F}^{(1)}(\phi) = F^{(1)}(\phi) + (D_{\omega}H^{(1)})(\phi)  = \sum_{\langle k, \omega\rangle = 0} A_{j,k} e^{i\langle k, \phi\rangle}  \, .$$  Now consider  
 the parameter-family of vector fields 
 $$F(\phi, \varepsilon)=\omega+\varepsilon F^{(1)}(\phi) + \varepsilon^2R^{(2)}(\phi, \varepsilon)\ \mbox{on} \  \mathbb{T}^n\, , $$  and assume, in addition to \eqref{assumptionsconvergence}, that $F^{(1)}\in C^2(\mathbb{T}^n), H^{(1)}\in C^3(\mathbb{T}^n)$ and $R^{(2)}\in C^2(\mathbb{T}^n\times [-\varepsilon_0,\varepsilon_0])$. 
 Then 
 \begin{itemize}
     \item[{\it i)}] the map 
 $(\phi, \varepsilon)\mapsto 
 \Gamma^{\varepsilon}_{H^{(1)}}(\phi)\ \mbox{from}\ \mathbb{T}^n\times [-\varepsilon_0,\varepsilon_0] \ \mbox{to}\  \mathbb{T}^n$ is $C^3$, and 
 \begin{equation}
     \label{eq:flowestimate}
 \|\Gamma_{H^{(1)}}^{\varepsilon}(\phi)_j -\phi_j \| \leq \varepsilon \|H_j^{(1)}\|_0 \, .
 \end{equation}
 \item[{\it ii)}]
  the map $(\phi, \varepsilon) \mapsto (\Gamma^{\varepsilon}_{H^{(1)}})_*F(\phi)$ from $\mathbb{T}^n\times [-\varepsilon_0,\varepsilon_0]$ to $\mathbb{R}^n$ is $C^2$, and 
\begin{equation}\label{eq:firstexpansion}
(\Gamma^{\varepsilon}_{H^{(1)}})_*F(\phi) = \omega + \varepsilon \overline{F}^{(1)}(\phi) + \varepsilon^2\overline{R}^{(2)}(\phi, \varepsilon)\, , 
\end{equation}
    for a bounded continuous vector field $\overline{R}^{(2)}: \mathbb{T}^n\times[-\varepsilon_0, \varepsilon_0]\to \mathbb{R}^n$.
\end{itemize}

\end{theorem}
\begin{proof}
    The  assumptions in \eqref{assumptionsconvergence} imply that the Fourier series for $F^{(1)}$ and $H^{(1)}$ given in \eqref{eq:F1expansion} and \eqref{eq:firstordernfformal} are  absolutely convergent, so that the partial Fourier sums are uniformly convergent. As an immediate consequence, their limits $F^{(1)}$ and $H^{(1)}$ are continuous. 
    Because $$\sum_{\langle k, \omega\rangle = 0 } |A_{j,k}^{(1)}| \leq \sum_{k\in\mathbb{Z}^n} |A_{j,k}^{(1)}|\, ,$$ 
    the same is automatically true for $\overline{F}^{(1)}$. By definition, $$D_{\omega}H^{(1)}(\phi)=-\sum_{\langle k, \omega\rangle \neq 0} A_{j,k}^{(1)} e^{i\langle k, \phi\rangle} = \overline{F}^{(1)}(\phi) - F^{(1)}(\phi)\, ,$$ 
    so also $D_{\omega}H^{(1)}$ is absolutely convergent and continuous. To summarize: \eqref{assumptionsconvergence} implies that the homological equation $\overline{F}^{(1)}(\phi) = F^{(1)}(\phi) + (D_{\omega}H^{(1)})(\phi)$ holds as an equality of continuous functions. 
    
    The remaining statements of the theorem follow immediately from the results in Appendix  \ref{app:Lie} and 
\ref{sec:normalformappendix}. Indeed, $\Gamma_{H^{(1)}}^{\varepsilon}$ is well-defined because $H^{(1)}\in C^3(\mathbb{T}^n)$ is Lipschitz continuous, and formula \eqref{eq:flowestimate} was already proved in Remark \ref{rem:flowremark}.
    It was also shown in  Appendix \ref{app:Lie} that when $(\phi, \varepsilon)\mapsto F(\phi, \varepsilon)$ is $C^2$ and $(\phi, \varepsilon)\mapsto H^{(1)}(\phi, \varepsilon)$ is $C^3$, then $(\phi, \varepsilon)\mapsto \Gamma_{H^{(1)}}^{\varepsilon}(\phi)$ is $C^3$, and therefore $(\phi, \varepsilon)\mapsto (\Gamma_{H^{(1)}}^{\varepsilon})_*F(\phi)$ is $C^2$.   Formula \eqref{eq:firstexpansion} now follows from Taylor expansion of $(\Gamma^{\varepsilon}_{H^{(1)}})_*F(\phi)$ to $\varepsilon$.  It was shown in Appendix \ref{sec:normalformappendix} that the order-$\varepsilon$ term of this expansion is indeed given by $\overline{F}^{(1)}$.
\end{proof}
\section{The second order normal form}\label{app:2ndorderNF}
We now calculate the second order normal form of equation \eqref{main} {\it under the assumption that the first order normal form of $F$ vanishes}. We do this only formally, not considering the convergence of any of the Fourier series involved in the computations. Recall that the second order term  $\overline{F}_2$ in the second order normal form of $F$ is given by formula  \eqref{eq:secondterm}. This expression considerably simplifies when it so happens that $\overline{F}^{(1)} =0$, as this means that $H^{(1)}$ fully solves the  homological equation \eqref{eq:firsthomologicalequation}, i.e., that $[H^{(1)}, \omega] = - {\rm ad}_{\omega}(H^{(1)}) = - F^{(1)}$. This implies that \eqref{eq:secondterm} simplifies to
\begin{equation}\label{eq:F2barsimplified}
\overline{F}^{(2)} = F^{(2)} + \frac{1}{2} [H^{(1)}, F^{(1)}] - {\rm ad}_{\omega}(H^{(2)})\, ,
\end{equation}
where $F^{(1)}$ and $H^{(1)}$ are as given in \eqref{eq:F1expansion}, \eqref{eq:H1expansion} and \eqref{eq:H1expansionexplicit}, and $H^{(2)}$ is chosen to cancel the nonresonant terms $F^{(2)} + \frac{1}{2} [H^{(1)}, F^{(1)}]$, so that \eqref{eq:F2barsimplified} contains only resonant terms.   The following result gives the formal Fourier expansion of $\overline{F}^{(2)}$.

\begin{theorem} \label{second_order_nf_thm}
  Assume that 
  $$F^{(1)}(\phi)=\sum_{k\in\mathbb{Z}^n} A^{(1)}_{j,k}e^{i\langle k,\phi\rangle}\ \mbox{and} \ F^{(2)}(\phi)=\sum_{k\in\mathbb{Z}^n} A^{(2)}_{j,k}e^{i\langle k,\phi\rangle}\, ,$$ 
  and that $\overline{F}^{(1)}=0$, i.e., that $A_{j,k}^{(1)}=0$ for all $1\leq j \leq n$, and all $k\in \mathbb{Z}^n$ that satisfy $\langle\omega,k\rangle=0$. Then
    \begin{align} \label{eq:F2barexpression}
        \overline{F}_j^{(2)}(\phi) = \sum_{\substack{\innerproduct{k+l}{\omega} = 0 \\\innerproduct{k}{\omega} \neq 0}} \sum_{m=1}^n\frac{ A_{j, k}^{(1)} k_mA_{m, l}^{(1)}}{\innerproduct{l}{\omega}} e^{i\innerproduct{k+l}{\phi} }  + 
 \sum_{\innerproduct{k}{\omega}=0} A_{j, k}^{(2)} e^{i \langle k, \phi \rangle} \, .
    \end{align}
\end{theorem}
\begin{proof}
The Lie bracket between $H^{(1)}$ and $F^{(1)}$ is given in components by 
$$[H^{(1)},F^{(1)}]_j = DH_j^{(1)}\cdot F^{(1)} - DF_j^{(1)}\cdot H^{(1)}\, . $$
We first compute the derivative of $H_j^{(1)}$ in the direction of $F^{(1)}$:
\begin{align*}
D H_j^{(1)}(\phi) \cdot F^{(1)}(\phi) & =  
 {\frac{d}{ds}}\Bigg |_{s=0} \sum_{k\in\mathbb{Z}^n}  H_{j,k}^{(1)} e^{i\innerproduct{k}{\phi+s F^{(1)}(\phi)}} \\
&= 
  \sum_{k\in\mathbb{Z}^n}  H_{j,k}^{(1)} \sum_{m=1}^n i k_m F_m^{(1)}(\phi) e^{i\innerproduct{k}{\phi}} \\
&=  \sum_{k\in\mathbb{Z}^n}  H_{j,k}^{(1)} \sum_{m=1}^n i k_m \sum_{l\in\mathbb{Z}^n} A_{m,l}^{(1)}e^{i\innerproduct{l}{\phi}}   e^{i\innerproduct{k}{\phi}} \\
&= 
  \sum_{k,l\in\mathbb{Z}^n}\sum_{m=1}^n i H_{j,k}^{(1)}   k_m     A_{m,l}^{(1)}e^{i\innerproduct{k+l}{\phi}} \, .
\end{align*}
An analogous computation gives that
\begin{align*}
D F_j^{(1)}(\phi) \cdot H^{(1)}(\phi) & = \sum_{k,l\in\mathbb{Z}^n}\sum_{m=1}^n i A_{j,k}^{(1)}   k_m     H_{m,l}^{(1)}e^{i\innerproduct{k+l}{\phi}}\, , 
\end{align*}
and combining this we find
\begin{equation}\label{eq:HFfirstexpression}
   \frac{1}{2} [H^{(1)},F^{(1)}]_j (\phi) =  \frac{1}{2} \sum_{k,l\in\mathbb{Z}^n}\sum_{m=1}^n i\left(  H_{j,k}^{(1)}   k_m     A_{m,l}^{(1)} - A_{j,k}^{(1)}   k_m     H_{m,l}^{(1)} \right) e^{i\innerproduct{k+l}{\phi}}  \, .
\end{equation}
Using that $H_{j,k}^{(1)} = -\frac{A_{j,k}^{(1)}}{i\langle k, \omega\rangle}$ and hence that $H^{(1)}_{m,l} = -\frac{A_{m,l}^{(1)}}{i\langle l, \omega\rangle}$  - see equation \eqref{eq:H1expansionexplicit} -   the right hand side of \eqref{eq:HFfirstexpression} is equal to
\begin{equation}\label{eq:HFsecondexpression}
   \frac{1}{2} \sum_{k,l\in\mathbb{Z}^n}\sum_{m=1}^n A_{j,k}^{(1)}k_m A_{m,l}^{(1)} \left(   
   \frac{1}{\langle l, \omega\rangle} 
   - \frac{1}{\langle k, \omega\rangle}   
    \right) e^{i\innerproduct{k+l}{\phi}}  \, .
\end{equation} 
Next, recall that the vector field $H^{(2)}$ is chosen so that it cancels all nonresonant terms in $F^{(2)} + \frac{1}{2}[H^{(1)}, F^{(1)}]$. This means that $\overline{F}^{(2)}$ is equal to the resonant part of $F^{(2)} + \frac{1}{2}[H^{(1)}, F^{(1)}]$. 
It is clear that a term in the sum \eqref{eq:HFsecondexpression} is resonant precisely when $\langle \omega, k+l\rangle = 0$, that is, when $\langle \omega, k\rangle = - \langle \omega, l\rangle$. For such $k$ and $l$ the coefficients in \eqref{eq:HFsecondexpression} simplify. Indeed, 
$$\frac{1}{2}A_{j,k}^{(1)}k_m A_{m,l}^{(1)} \left(   
   \frac{1}{\langle l, \omega\rangle} 
   - \frac{1}{\langle k, \omega\rangle}   
    \right) =  \frac{A_{j,k}^{(1)}k_m A_{m,l}^{(1)}}{\langle l, \omega\rangle}  \ \  \mbox{when}\ \langle k+l, \omega\rangle = 0\, .$$
We also remark that a term in $F^{(2)}_j(\phi) = \sum_k A^{(2)}_{j,k}e^{i\langle k, \phi\rangle}$ is resonant if and only if $\langle \omega, k\rangle = 0$.
This proves that $\overline{F}_j^{(2)}(\phi)$ is exactly given as in formula \eqref{eq:F2barexpression},  and completes the proof of the theorem. 
\end{proof}

\noindent We refrain from formulating conditions  that guarantee the smoothness of the Lie transformations $\Gamma_{H^{(1)}}^{\varepsilon}$ and $\Gamma_{H^{(2)}}^{\varepsilon^2}$ that bring $F$ into normal form to second order. We do not exploit such conditions in the main body of the paper.
\begin{remark}
    The second order part $\overline{F}_j^{(2)}$ of the $j$-th component of the normal form contains a nonzero resonant term 
    \begin{equation}\label{eq:resonantinteractionterm}
        \frac{ A_{j, k}^{(1)} k_mA_{m, l}^{(1)}}{\innerproduct{l}{\omega}} e^{i\innerproduct{k+l}{\phi} }
     \end{equation}
     precisely when 
     \begin{itemize}
         \item[{\it i)}] $F_j^{(1)}$ contains a nonzero nonresonant term $A_{j,k}^{(1)}e^{i\langle k, \phi\rangle}$;
         \item[{\it ii)}] It holds that $k_m\neq 0$, i.e., the term in {\it i)} depends nontrivially on $\phi_m$; 
         \item[{\it iii)}] 
       $F_m^{(1)}$ contains a nonzero nonresonant term $ A_{m, l}^{(1)} e^{i\innerproduct{l}{\phi}}$;
        \item[{\it iv)}]  The resonance condition $\langle \omega, k\rangle = - \langle \omega, l\rangle $ holds;
     \end{itemize} 
           In other words, \eqref{eq:resonantinteractionterm} represents an effective indirect resonant interaction between those nodes $i$ that satisfy $(k+l)_i\neq 0$, via node $m$, to node $j$.  
\end{remark}
\begin{remark}
Let us define the support of a Fourier label $k\in \mathbb{Z}^n$ by 
    $${\rm supp}(k) = \{ 1\leq i \leq n\, |\, k_i\neq 0\} \subset \{1, \ldots, n\} \, .$$
    It contains the nodes on which the Fourier mode $e^{i\langle k, \phi\rangle}$ depends nontrivially. Note that ${\rm supp}(k+l)\subset {\rm supp}(k) \cup {\rm supp}(l)$. 
    
    Let us now assume that $A_{j,k}^{(1)} k_m A_{m,l}^{(1)} \neq 0$ and that $k$ and $l$ both represent an ``edge'' or ``dyadic interaction''. This means that ${\rm supp}(k)=\{j, j_1\}$ for some node $j_1\neq j$, that ${\rm supp}(l)=\{m, j_2\}$ for some node $j_2\neq m$, and that $m\in \{j, j_1\}$. 
    For simplicity, let us furthermore assume that $\{m\} = \{j, j_1\} \cap \{m, j_2\}$ consists of exactly one element. Then the support of $k+l$ has three elements, and the Fourier mode $e^{i\langle k+l,\phi\rangle}$ signifies a ``triadic interaction'' between the nodes in the supports of $k$ and $l$. This can occur in two distinct ways:
\begin{enumerate}
\item \textbf{Feedforward motif}: If $j \neq m = j_1$, we have ${\rm supp}(k+l) = \{j, m, j_2\}$. This triadic interaction 
represents a feedforward  motif with 3 nodes, as illustrated in Figure \ref{fig:feedforward2}. It is formed by concatenating the  edge from  $j_2$ to $m$ and the edge from $m=j_1$ to $j$, as shown in Figure \ref{fig:feedforward1}.

\item \textbf{Wedge motif} : If $j = m \neq j_2$, then ${\rm supp}(k+l) = \{j, j_1, j_2\}$. This triadic interaction 
represents a wedge motif with 3 nodes, as shown in Figure \ref{fig:wedge2}. It is formed by  merging the edges from $j_1$ to $j$ and from $j_2$ to $j$, as shown in Figure \ref{fig:wedge1}.

\end{enumerate}
\end{remark}

\begin{figure}[h]
    \centering
    \begin{minipage}[b]{0.45\textwidth}
        \centering
        \resizebox{5cm}{!}{
\usetikzlibrary{decorations.pathmorphing, shadows}
\begin{tikzpicture}[->, >=stealth, thick, every node/.style={circle, draw=none, minimum size=1.5cm, font=\bfseries\Large},scale=1, transform shape]
  \definecolor{lightblue}{RGB}{203, 214, 230}
  \definecolor{turquoise}{RGB}{64, 224, 208}
  \definecolor{lightturquoise}{RGB}{194, 255, 255}  
  \definecolor{brightyellow}{RGB}{255, 255, 102}    
  \definecolor{lightyellow}{RGB}{255, 255, 204}     
  
  \node[fill=lightblue,drop shadow] (1) at (0, 0) {$j_2$};
  \node[fill=turquoise,drop shadow] (2) at (3.5, 0) {$m=j_1$};  
  \node[fill=brightyellow,drop shadow] (3) at (7, 0) {$j$};  

  \path[thick] (1) edge node[above, draw=none, font=\bfseries\large] {$A_{m, l}^{(1)}$} (2);
  \path[thick] (2) edge node[above, draw=none, font=\bfseries\large] {$A_{j, k}^{(1)}$} (3);
  
  \draw[thick, fill=lightturquoise, fill opacity=0.3, draw=lightturquoise, decorate, decoration={snake, amplitude=1mm, segment length=12mm}]
        (1.75, -0.3) ellipse (3cm and 1.5cm);
  
  \draw[thick, fill=lightyellow, fill opacity=0.3, draw=lightyellow, decorate, decoration={snake, amplitude=1mm, segment length=12mm}]
        (5.25, -0.3) ellipse (3cm and 1.5cm);
  
\end{tikzpicture}}
        \caption{Feedforward motif in the original first order equation}
        \label{fig:feedforward1}
    \end{minipage}
    \hfill
    \begin{minipage}[b]{0.45\textwidth}
        \centering
        \resizebox{5cm}{!}{
\usetikzlibrary{decorations.pathmorphing, shadows}
\begin{tikzpicture}[->, >=stealth, thick, every node/.style={circle, draw=none, minimum size=1.5cm, font=\bfseries\Large},scale=1, transform shape]
  \definecolor{lightblue}{RGB}{203, 214, 230}
  \definecolor{turquoise}{RGB}{64, 224, 208}
  \definecolor{brightyellow}{RGB}{255, 255, 102}    
  \definecolor{limegreen}{RGB}{204, 255, 153}           
  
  \node[fill=lightblue,drop shadow] (1) at (0, 0) {$j_2$};
  \node[fill=turquoise,drop shadow] (2) at (3.5, 0) {$m=j_1$};  
  \node[fill=brightyellow,drop shadow] (3) at (7, 0) {$j$};  

  \path[thick] (1) edge (2);
  \path[thick] (2) edge (3);
  
  \draw[thick, fill=limegreen, fill opacity=0.2, draw=limegreen, decorate, decoration={snake, amplitude=1mm, segment length=12mm}]
        (3.5, -0.3) ellipse (5cm and 2.0cm);

  \node[font=\bfseries\large] at (7.8, -2.0) {$\frac{ A_{j, k}^{(1)}k_m A_{m, l}^{(1)}}{\langle l, \omega \rangle}$};
  
\end{tikzpicture}}
        \caption{Resulting hyperedge in the second order normal form}
        \label{fig:feedforward2}
    \end{minipage}
    
    \vspace{1cm}

    \begin{minipage}[b]{0.45\textwidth}
        \centering
        \resizebox{5cm}{!}{
\usetikzlibrary{decorations.pathmorphing, shadows,decorations.markings}

\begin{tikzpicture}[->, >=stealth, thick, every node/.style={circle, draw=none, minimum size=1.5cm, font=\bfseries\Large, text=black}]

  \definecolor{lightblue}{RGB}{173, 216, 230}   
  \definecolor{brightyellow}{RGB}{255, 255, 153} 
  \definecolor{turquoise}{RGB}{72, 209, 204}    
  \definecolor{yellowellipse}{RGB}{255, 255, 102}  

  \node[fill=lightblue, text=black,drop shadow] (1) at (-2, 1) {$j_1$};             
  \node[fill=brightyellow, text=black,drop shadow] (2) at (0, -0.7) {$j=m$};       
  \node[fill=turquoise, text=black,drop shadow] (3) at (2, 1) {$j_2$};              

  \path[thick, draw=black] 
    (1) edge[bend right=20] node[above, midway, sloped, font=\bfseries\Large] {$A_{j, k}^{(1)}$} (2);  
  
  \path[thick, draw=black] 
    (3) edge[bend left=20] node[above, midway, sloped, font=\bfseries\Large] {$A_{j, l}^{(1)}$} (2);  

  \draw[ultra thick, fill=yellowellipse, fill opacity=0.2, draw=yellowellipse, rotate around={-30:(-1,0.3)}, decorate, decoration={coil, amplitude=1mm, segment length=12mm}]
        (-1, 0.3) ellipse (3cm and 1.5cm);
  
  \draw[ultra thick, fill=yellowellipse, fill opacity=0.1, draw=yellowellipse, rotate around={30:(1,0.3)}, decorate, decoration={coil, amplitude=1mm, segment length=12mm}]
        (1, 0.3) ellipse (3cm and 1.5cm);

\end{tikzpicture}}
        \caption{Wedge motif in the original first order equation}
        \label{fig:wedge1}
    \end{minipage}
    \hfill
    \begin{minipage}[b]{0.45\textwidth}
        \centering
        \resizebox{5cm}{!}{
\usetikzlibrary{decorations.pathmorphing, shadows, decorations.markings,}

\begin{tikzpicture}[->, >=stealth, thick, every node/.style={circle, draw=none, minimum size=1.5cm, font=\bfseries\Large, text=black}]

  \definecolor{lightblue}{RGB}{173, 216, 230}   
  \definecolor{brightyellow}{RGB}{255, 255, 153} 
  \definecolor{turquoise}{RGB}{72, 209, 204}    
  \definecolor{limegreen}{RGB}{204, 255, 153}   

  \node[fill=lightblue, color=lightblue, text=black, drop shadow] (1) at (-2, 1) {$j_1$};             
  \node[fill=brightyellow,color=yellow, text=black,drop shadow, font=\bfseries\Large] (2) at (0, -0.7) {$j=m$};  
  \node[fill=turquoise,color=turquoise, text=black, drop shadow] (3) at (2, 1) {$j_2$};              

  \path[thick, draw=black,, decoration={snake, amplitude=0.5mm}] (1) edge[bend right=20] (2);  
  \path[thick, draw=black,, decoration={snake, amplitude=0.5mm}] (3) edge[bend left=20] (2);  

 \draw[ultra thick, fill=limegreen, fill opacity=0.3, draw=limegreen,  decorate, decoration={coil, amplitude=1mm, segment length=12mm}]
       (0, 0.3) ellipse (3.2cm and 2.5cm);

  \node[font=\bfseries\Large, color=black] at (4.2, -1.8) {$\frac{ A_{j, k}^{(1)}k_j A_{j, l}^{(1)}}{\langle l, \omega \rangle}$};
  
\end{tikzpicture}} 
        \caption{Resulting hyperedge in the second order normal form}
        \label{fig:wedge2}
    \end{minipage}    
    \label{fig:networks}
\end{figure}
\bibliographystyle{plain}
\bibliography{references}
\end{document}